\documentclass[11pt,a4paper]{article}
\usepackage[utf8]{inputenc}
\usepackage[T1]{fontenc}
\usepackage[english]{babel}
\usepackage[left=2.1cm,right=2.1cm,top=2.2cm,bottom=2.2cm]{geometry}

\usepackage{amsmath,amsthm,amssymb}
\usepackage{mathtools}
\usepackage{newtxmath}
\usepackage[scaled=0.92]{helvet}

\DeclareSymbolFont{operators}{T1}{qhv}{m}{n}

\usepackage{booktabs}
\usepackage{tabularx}
\usepackage{tikz}
\usetikzlibrary{arrows.meta, positioning, calc, decorations.pathreplacing,
  shapes.geometric, shapes.misc}
\usepackage[activate={true,nocompatibility},final,tracking=true,
  kerning=true,spacing=true,factor=1100,stretch=15,shrink=15,
  expansion=true,protrusion=true]{microtype}
 
\usepackage{xcolor}
\usepackage{colortbl}
\usepackage{titlesec}
\usepackage{hyperref}
\usepackage{url}
\providecommand{\cref}[1]{\ref{#1}}
\providecommand{\Cref}{\cref}
\usepackage[font=small,labelfont=bf,skip=4pt]{caption}
\usepackage{enumitem}
\makeatletter
\providecommand{\needspace}[1]{%
  \par
  \ifdim\pagegoal=\maxdimen\else
    \dimen@=\pagegoal \advance\dimen@ by -\pagetotal
    \ifdim\dimen@<#1\relax\newpage\fi
  \fi}
\makeatother
\providecommand{\FloatBarrier}{\par\vspace*{0pt}%
  \penalty -100\relax
  \ifvmode\else\par\fi}
\allowdisplaybreaks[1]
\setlist[itemize]{noitemsep, topsep=2pt, partopsep=0pt,
                  parsep=0pt, leftmargin=1.5em}
\setlist[description]{font=\normalfont\itshape, leftmargin=2em,
                      itemsep=2pt}

\definecolor{cybiontblue}{RGB}{0, 70, 140}
\definecolor{darkgray}{RGB}{60, 60, 60}

\titleformat{\section}
  {\Large\bfseries\color{cybiontblue}}
  {\thesection}{1em}{}
\titleformat{\subsection}
  {\large\bfseries\color{darkgray}}
  {\thesubsection}{1em}{}
\titlespacing*{\section}{0pt}{1.3ex plus 0.4ex minus .3ex}{0.6ex plus .1ex minus .1ex}
\titlespacing*{\subsection}{0pt}{0.9ex plus 0.3ex minus .2ex}{0.4ex plus .1ex minus .1ex}

\hypersetup{
  colorlinks=true,
  linkcolor=cybiontblue,
  citecolor=cybiontblue,
  urlcolor=cybiontblue,
  destlabel=true,
  hypertexnames=false,
  pdftitle={Mathematical Foundations for Peer-to-Peer Lattice Computation},
  pdfauthor={Danil Gorinevski},
  pdfsubject={Measure-theoretic elevation, sparse-participation scaling, functorial admissibility, subcritical percolation bounds for peer-to-peer grid graph computation},
  pdfkeywords={optimal transport, Monge-Kantorovich, Wasserstein, graphon, Lawvere theory, proof-carrying code, percolation, parallel computing, grid graph, distributed reduction},
}

\newtheorem{proposition}{Proposition}

\newtheorem{lemma}{Lemma}
\newtheorem{corollary}{Corollary}

\newtheorem{remark}{Remark}

\numberwithin{equation}{section}

\title{%
  \vspace{-3em}%
  \textbf{Mathematical Foundations for\\
  Peer-to-Peer Lattice Computation}\\[0.3em]
  \large
  Five Conditional Bounds and Algebraic Criteria:
  Monge--Kantorovich, Sparse-Participation Scaling,
  Functorial Admissibility, Percolation, and Small-World
}
\author{%
  Danil Gorinevski\thanks{cybiont GmbH,
    Sch\"ubelbach, Switzerland.
    \texttt{research@cybiont.com}.}%
}
\date{2026}

\begin{document}
\maketitle

\begin{abstract}
We give structured proofs for five mathematical propositions that govern the behaviour of synchronous peer-to-peer computation on a finite grid graph embedded in $\mathbb{Z}^2$. Proposition 1 gives three lower bounds: a transport-work bound $\sum_i a_i \ell_i \geq W_1(\mu,\nu)$ attained by every shortest-path schedule; a completion-depth bound $D_{\min} \geq r_\mu$ in terms of the $\mu$-support radius, attained by non-congesting parallel shortest-path routing; and a compressive-reduction edge bound $|E'| \geq \mathrm{St}_G(\mathrm{supp}(\mu)\cup\{x_\star\})$ in terms of the graph-Steiner cost. A negative result establishes that the sink-trunk route-load functional under i.i.d.\ Bernoulli activation has variance $\Theta(f_{\text{act}}(1-f_{\text{act}})P^2)$ for corner-sink dimension-order routing, refuting any naive $O(f_{\text{act}}P^{3/2})$ concentration claim. Proposition 2 establishes, under the $\alpha$-$\beta$-$\gamma$ collective-communication model and a Mixture-of-Experts sparse-activation model, that the grid-to-cluster latency ratio improves monotonically as $f_{\text{act}}$ shrinks whenever the cluster's fixed message/control overhead dominates the grid's geometric constant. Proposition 3 identifies a sufficient algebraic criterion for schedule-independent reduction semantics: update rules decomposing into a local map and an abelian-monoid merge, expressed as a product-preserving functor from the Lawvere theory of commutative monoids into the hardware-state category. Proposition 4 bounds the conditional expected route length under i.i.d.\ site failure in the subcritical regime $\delta < p_c^{\text{site}}(\mathbb{Z}^2)$ by an additive detour term, using Aizenman--Barsky exponential cluster-size decay. Proposition 5 augments the grid graph with $k$ long-range shortcuts per node, collapsing the typical shortest-path length from $\Theta(\sqrt{P})$ to $O(\log P)$ under a mean-field (Erd\H{o}s--R\'enyi) universality argument --- rigorous for the $1$-D-ring base (Newman--Watts--Strogatz) and conjectural for the $2$-D-grid base.
\end{abstract}

\section{Problem statement}
\label{sec:intro}

Many applications reduce to the same geometric primitive: a
finite set of participants carries distributed state across a
planar graph, partial results must be combined toward a
designated origin, and only local neighbour-to-neighbour
communication is available within one unit of synchronous time.
The primitive appears in parallel computation with stationary
local state, in distributed aggregation over sensor or
message-passing networks, in logistics problems on gridded
warehouse or transport topologies, in wave-propagation and
finite-element simulations whose grid graph is the physical
domain itself, and in a range of other settings where state is
naturally distributed and coordination is naturally local.

Five questions recur at the foundation level of this primitive,
drawn from five distinct mathematical frameworks.  Each question
has been answered at the level of applied intuition for years,
and each has a proof sketch in the literature invoking the
relevant framework by name.  This note collects them into a
single structured derivation with explicit assumptions and a
citation at each external-theorem dependency.

\begin{itemize}
  \item[\bf P1.] Does the algorithmic lower bound on coordination
    depth admit a continuous measure-theoretic lift, so that
    routing on the grid graph can be compared against the
    Monge--Kantorovich optimal-transport cost of the underlying
    state distribution?
  \item[\bf P2.] When only a small fraction $f_{\mathrm{act}}$
    of the participants contribute to a given output, does
    the grid graph's advantage over a dense communication graph
    grow as $f_{\mathrm{act}}$ shrinks?
  \item[\bf P3.] Which workloads have schedule-independent
    reduction semantics, and under what scheduler assumptions
    do they admit a saturating grid-graph execution?  Is that
    schedule-independence criterion a formally-checkable
    property of the workload's update rule?
  \item[\bf P4.] Under random participant failure, how does the
    worst-case coordination latency degrade?  Is there a regime
    where the degradation is sub-catastrophic --- an additive
    detour term rather than a collapse?
  \item[\bf P5.] Does adding a sparse budget of long-range
    shortcuts to the grid graph materially change the
    foundations --- both the achievable diameter and the
    near-critical fault-tolerance behaviour?
\end{itemize}

These five questions sit at the foundation level of the
primitive because each one closes a design or analytical
question that practitioners ask, and that the applied literature
answers only by sketch or by intuition.  P1 asks whether
$\Theta(\sqrt{P})$ depth is the fundamental limit or whether
routing slack remains: scaling experiments alone cannot tell, so
we elevate the algorithmic depth lower bound to
optimal-transport semantics on the underlying state measure and
match it tightly.  P2 asks whether the grid's scaling advantage
holds when activation gets sparse: sparse-expert routing makes
this concrete, and we settle it monotonically under an explicit
latency model.  P3 asks when a workload is safe to fold in any
order: a schedule-independence criterion turns informal
judgement into a machine-checkable algebraic property, and we
identify the sufficient functorial structure.  P4 asks whether
latency degrades gracefully or collapses under random tile
failure: field-defect rates make this the design-margin
question, and we bound the conditional expected detour in the
subcritical regime by an additive term.  P5 asks what a sparse
long-range shortcut budget buys: we record the typical-distance
collapse from $\Theta(\sqrt{P})$ to $O(\log P)$ and the
near-critical closure that the pure-$\mathbb{Z}^2$ case does
not admit.  The questions are framed at a deliberately neutral
level of abstraction --- a finite connected graph, a finite set
of participants, a discrete measure on the graph, a synchronous
cycle model --- so the results instantiate across parallel
computation, distributed simulation, logistics, and sensor
networks by choice of particular graphs, cost functions, and
admissible update rules.

\section{Engineering context}
\label{sec:context}

The five propositions correspond to five challenges that arise
whenever a peer-to-peer grid graph is used as a coordinated
computational graph at scale.  Five cartoons fix the objects.

\subsection{Distributed state and its convergence}
When the state accessed per output is large relative to the
arithmetic per output, the bottleneck is state movement rather
than compute.  A peer-to-peer grid graph with \emph{stationary} local
state inverts the traffic: the payload travels, the state does
not.  Each participant $v \in V$ holds a fixed partition of the
global state; a lightweight payload traverses edges under the
constraint that at most one edge is crossed per synchronous
cycle.  The transport work required to converge to a designated origin
is bounded below by the mean graph distance of the initial state
distribution --- the $W_1$ Wasserstein distance on the Manhattan
metric --- and the completion depth is bounded below by the
support radius.  Figure~\ref{fig:grid-cartoon} fixes the
grid-graph geometry and the reduction target.  Proposition~1
makes both bounds precise and asks whether the grid graph
attains them.

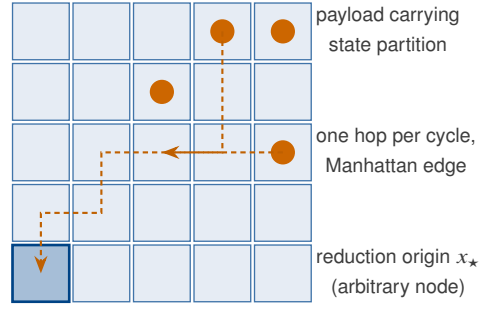
\begin{figure}[!tb]
\centering
\begin{tikzpicture}[
  meshnode/.style={rectangle, draw=cybiontblue!75, fill=cybiontblue!10,
    line width=0.5pt, minimum size=0.75cm, inner sep=0pt,
    font=\scriptsize},
  origin/.style={rectangle, draw=cybiontblue!95!black,
    fill=cybiontblue!35, line width=1pt,
    minimum size=0.75cm, inner sep=0pt,
    font=\scriptsize\bfseries},
  payload/.style={circle, draw=none, fill=orange!80!black,
    minimum size=3.2mm, inner sep=0pt},
  flow/.style={->, >=Stealth, draw=orange!75!black, line width=0.8pt,
    dash pattern=on 2pt off 1.5pt},
  lbl/.style={font={\scriptsize\linespread{1.18}\selectfont}, text=darkgray, align=center, inner sep=1pt}
]
\def\sz{0.80}
\foreach \i in {0,...,4}
  \foreach \j in {0,...,4}
    \node[meshnode] at (\i*\sz, \j*\sz) {};
\node[origin] at (0,0) {};
\node[payload] at (3*\sz, 4*\sz) {};
\node[payload] at (4*\sz, 2*\sz) {};
\node[payload] at (2*\sz, 3*\sz) {};
\node[payload] at (4*\sz, 4*\sz) {};
\draw[flow] (3*\sz, 4*\sz) -- (3*\sz, 3*\sz) -- (3*\sz, 2*\sz)
            -- (2*\sz, 2*\sz) -- (1*\sz, 2*\sz) -- (1*\sz, 1*\sz)
            -- (0, 1*\sz) -- (0, 0);
\draw[flow] (4*\sz, 2*\sz) -- (3*\sz, 2*\sz) -- (2*\sz, 2*\sz);
\node[lbl, anchor=west] at (4.5*\sz, 4*\sz)
  {payload carrying\\ state partition};
\node[lbl, anchor=west] at (4.5*\sz, 2*\sz)
  {one hop per cycle,\\ Manhattan edge};
\node[lbl, anchor=west] at (4.5*\sz, 0)
  {reduction origin $x_\star$\\ (arbitrary node)};
\end{tikzpicture}
\caption{The peer-to-peer grid graph object.  A finite graph $G = (V, E)$
embedded in $\mathbb{Z}^2$, edges between Manhattan-nearest
neighbours, each node $v \in V$ carrying stationary local state.
Payloads (orange) start distributed across nodes and must converge
at the origin $x_\star$ (dark corner).  In one synchronous cycle
each payload traverses at most one edge.  Proposition~1 asks for
the optimal-transport cost of this reduction; Proposition~3
describes which payload-update rules admit a saturating schedule;
Proposition~4 asks what happens when some nodes are dead.}
\label{fig:grid-cartoon}
\end{figure}

\subsection{Compositional reduction}
Once state is distributed across $P$ participants, any reduction
must pick an order in which partial results are combined.  On a
peer-to-peer grid graph the order is \emph{underdetermined}: there are
exponentially many trees that could route the reduction.  If the
per-participant update rule factors into a local map
$g(p, S_i)$ and a global merge $\oplus$ such that
$(A, \oplus, \mathbf{0})$ is an abelian monoid, \emph{every}
parenthesisation returns the same value.  The scheduler is free
to reshape the reduction at run-time, routing may deflect around
obstacles, partial results may merge at any convenient
participant --- none of these change the semantics.  The
abelian-monoid decomposition is a correctness discipline, not a
performance optimisation.  Figure~\ref{fig:fold-cartoon}
diagrams the tree-fold.  Proposition~3 gives a sufficient
algebraic criterion for workloads that admit it and shows the
criterion is in principle machine-checkable at compile time.

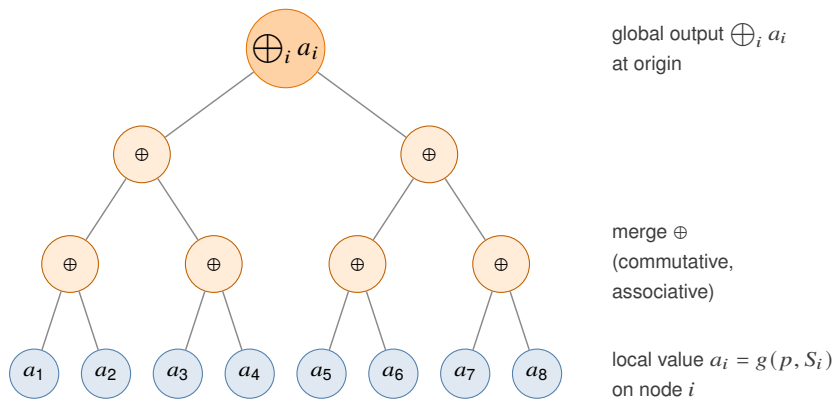
\begin{figure}[!b]
\centering
\begin{tikzpicture}[
  leaf/.style={circle, draw=cybiontblue!70, fill=cybiontblue!12,
    inner sep=0.5pt, minimum size=6.5mm, font=\scriptsize},
  ibd/.style={circle, draw=orange!75!black, fill=orange!15,
    inner sep=0.5pt, minimum size=7.5mm, font=\scriptsize\bfseries},
  rt/.style={circle, draw=orange!90!black, fill=orange!35,
    inner sep=0.5pt, minimum size=9mm, font=\small\bfseries},
  ed/.style={draw=darkgray!60, line width=0.55pt},
  lbl/.style={font={\scriptsize\linespread{1.18}\selectfont}, text=darkgray, align=left, inner sep=1pt}
]
\def\sp{0.95}
\foreach \i/\lbl in {0/$a_1$,1/$a_2$,2/$a_3$,3/$a_4$,
                     4/$a_5$,5/$a_6$,6/$a_7$,7/$a_8$} {
  \node[leaf] (l\i) at (\i*\sp, 0) {\lbl};
}
\node[ibd] (m0) at (0.5*\sp, 1.45) {$\oplus$};
\node[ibd] (m1) at (2.5*\sp, 1.45) {$\oplus$};
\node[ibd] (m2) at (4.5*\sp, 1.45) {$\oplus$};
\node[ibd] (m3) at (6.5*\sp, 1.45) {$\oplus$};
\foreach \i in {0,1} { \draw[ed] (l\i) -- (m0); }
\foreach \i in {2,3} { \draw[ed] (l\i) -- (m1); }
\foreach \i in {4,5} { \draw[ed] (l\i) -- (m2); }
\foreach \i in {6,7} { \draw[ed] (l\i) -- (m3); }
\node[ibd] (n0) at (1.5*\sp, 2.9) {$\oplus$};
\node[ibd] (n1) at (5.5*\sp, 2.9) {$\oplus$};
\draw[ed] (m0) -- (n0); \draw[ed] (m1) -- (n0);
\draw[ed] (m2) -- (n1); \draw[ed] (m3) -- (n1);
\node[rt] (rt) at (3.5*\sp, 4.3) {$\bigoplus_{i} a_i$};
\draw[ed] (n0) -- (rt); \draw[ed] (n1) -- (rt);

\node[lbl, anchor=west] at (8.0*\sp, 0.0)
  {local value $a_i = g(p, S_i)$\\ on node $i$};
\node[lbl, anchor=west] at (8.0*\sp, 1.45)
  {merge $\oplus$ \\ (commutative,\\ associative)};
\node[lbl, anchor=west] at (8.0*\sp, 4.3)
  {global output $\bigoplus_i a_i$\\ at origin};
\end{tikzpicture}
\caption{The abelian-monoid fold.  Each node $i$ produces a local
value $a_i = g(p, S_i)$ from its stationary state $S_i$ and the
traversing payload $p$.  A commutative associative merge
operator $\oplus$ combines values level by level; because
$\oplus$ makes $(A, \oplus, 0)$ an abelian monoid, the final
output is independent of the pairing order.  Tree-fold depth is
$\lceil \log_2 P \rceil$ on an $\oplus$-compatible reduction
graph; Proposition~2 and Proposition~3 rest on this structure.}
\label{fig:fold-cartoon}
\end{figure}

\subsection{Sparse per-output participation}
In many workloads only a fraction $f_{\mathrm{act}}$ of the
participants contribute to any particular output; the remaining
participants hold resident state but do not compute on that
output, and the active pattern may differ from output to output.
In the latency model analysed below, the dominant grid-side
coordination term is geometric (scaling with $\sqrt{P}$) and
independent of $f_{\mathrm{act}}$; sparse activation then saves
local compute without changing this geometric term.  The
corresponding dense-communication-graph model has a fixed
per-collective overhead that does not vanish as
$f_{\mathrm{act}} \to 0$.  Figure~\ref{fig:sparse-cartoon}
illustrates the active/resident split.  Proposition~2
formalises the asymmetry under these model assumptions.

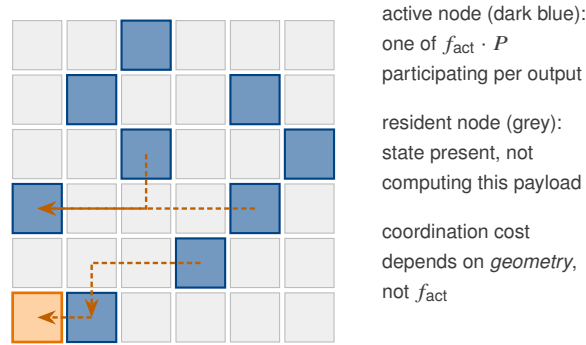
\begin{figure}[!t]
\centering
\begin{tikzpicture}[
  resident/.style={rectangle, draw=darkgray!35, fill=darkgray!8,
    line width=0.4pt, minimum size=0.65cm, inner sep=0pt},
  active/.style={rectangle, draw=cybiontblue!95!black,
    fill=cybiontblue!55, line width=0.8pt,
    minimum size=0.65cm, inner sep=0pt},
  origin/.style={rectangle, draw=orange!90!black,
    fill=orange!35, line width=1pt,
    minimum size=0.65cm, inner sep=0pt},
  route/.style={->, >=Stealth, draw=orange!75!black, line width=0.9pt,
    dash pattern=on 2pt off 1.2pt},
  lbl/.style={font={\scriptsize\linespread{1.18}\selectfont}, text=darkgray, align=left, inner sep=1pt}
]
\def\sz{0.72}
\foreach \i in {0,...,5}
  \foreach \j in {0,...,5}
    \node[resident] at (\i*\sz, \j*\sz) {};
\foreach \i/\j in {1/0, 3/1, 0/2, 4/2, 2/3, 5/3, 1/4, 4/4, 2/5} {
  \node[active] at (\i*\sz, \j*\sz) {};
}
\node[origin] at (0,0) {};
\draw[route] (1*\sz, 0) -- (0, 0);
\draw[route] (3*\sz, 1*\sz) -- (1*\sz, 1*\sz) -- (1*\sz, 0);
\draw[route] (4*\sz, 2*\sz) -- (0*\sz, 2*\sz);
\draw[route] (2*\sz, 3*\sz) -- (2*\sz, 2*\sz) -- (0, 2*\sz);
\node[lbl, anchor=west] at (6*\sz + 0.2, 5*\sz)
  {active node (dark blue):\\
   one of $f_{\mathrm{act}}\cdot P$\\
   participating per output};
\node[lbl, anchor=west] at (6*\sz + 0.2, 3*\sz)
  {resident node (grey):\\
   state present, not\\
   computing this payload};
\node[lbl, anchor=west] at (6*\sz + 0.2, 1*\sz)
  {coordination cost\\
   depends on \emph{geometry},\\
   not $f_{\mathrm{act}}$};
\end{tikzpicture}
\caption{Sparse-activation pattern.  Only a fraction
$f_{\mathrm{act}} \cdot P$ of the $P$ nodes (dark blue) contribute
to the per-output fold; the remaining nodes (grey) hold resident
state but do not compute.  On the grid graph, coordination cost scales
with the \emph{geometric} diameter $\sqrt{P}$ of the route graph,
independent of which fraction of nodes happens to be active ---
this is why the sparse-regime advantage in Proposition~2 arises.
On a dense-cluster graph, the analogous coordination cost scales
with the \emph{participant count} $N_{\mathrm{gpu}}$ and does not
vanish as $f_{\mathrm{act}} \to 0$.}
\label{fig:sparse-cartoon}
\end{figure}

\subsection{Resilience under random participant failure}
A graph of $P = 10^4$ participants is a reliability problem even
before it is a performance problem: at any realistic individual
failure rate, some participant is absent during almost every
large reduction.  Most parallel runtimes respond by aborting and
restarting; the grid graph can instead detour a payload around an
isolated failed participant at the cost of exactly $+2$ edges, and
the cumulative detour across a subcritical random-failure field
grows only additively in the number of failed participants.
Figure~\ref{fig:deflection-cartoon} shows the isolated-obstacle
detour geometry.  Proposition~4 bounds the expected degradation
and identifies the regime in which it remains additive.

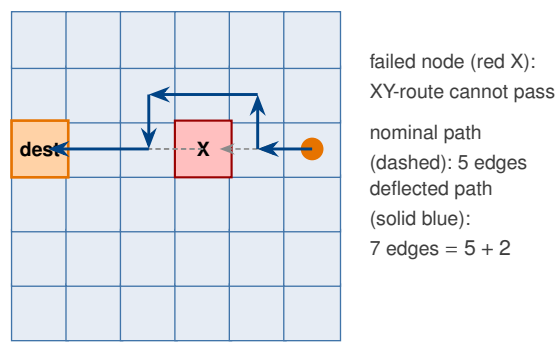
\begin{figure}[!b]
\centering
\begin{tikzpicture}[
  meshnode/.style={rectangle, draw=cybiontblue!75, fill=cybiontblue!10,
    line width=0.5pt, minimum size=0.75cm, inner sep=0pt},
  dead/.style={rectangle, draw=red!60!black, fill=red!25,
    line width=0.8pt, minimum size=0.75cm, inner sep=0pt,
    font=\scriptsize\bfseries},
  origin/.style={rectangle, draw=orange!90!black,
    fill=orange!35, line width=1pt,
    minimum size=0.75cm, inner sep=0pt,
    font=\scriptsize\bfseries},
  src/.style={circle, draw=none, fill=orange!90!black,
    minimum size=3mm, inner sep=0pt},
  nominal/.style={->, >=Stealth, draw=darkgray!60, line width=0.7pt,
    dash pattern=on 2pt off 1.2pt},
  detour/.style={->, >=Stealth, draw=cybiontblue!95!black,
                line width=1.1pt},
  lbl/.style={font={\scriptsize\linespread{1.18}\selectfont}, text=darkgray, align=left, inner sep=1pt}
]
\def\sz{0.72}
\foreach \i in {0,...,5}
  \foreach \j in {0,...,5}
    \node[meshnode] at (\i*\sz, \j*\sz) {};
\node[dead] at (3*\sz, 3*\sz) {\textsf{X}};
\node[origin] at (0, 3*\sz) {\textsf{dest}};
\node[src] at (5*\sz, 3*\sz) {};
\draw[nominal] (5*\sz, 3*\sz) -- (4*\sz, 3*\sz);
\draw[nominal, line width=0.6pt] (4*\sz, 3*\sz) -- (3.3*\sz, 3*\sz);
\draw[nominal] (3*\sz, 3*\sz) -- (0.2*\sz, 3*\sz);
\draw[detour] (5*\sz, 3*\sz) -- (4*\sz, 3*\sz);
\draw[detour] (4*\sz, 3*\sz) -- (4*\sz, 4*\sz);
\draw[detour] (4*\sz, 4*\sz) -- (2*\sz, 4*\sz);
\draw[detour] (2*\sz, 4*\sz) -- (2*\sz, 3*\sz);
\draw[detour] (2*\sz, 3*\sz) -- (0.15*\sz, 3*\sz);
\node[lbl, anchor=west] at (6.0*\sz, 4.3*\sz)
  {failed node (red X):\\
   XY-route cannot pass};
\node[lbl, anchor=west] at (6.0*\sz, 3.0*\sz)
  {nominal path\\
   (dashed): 5 edges};
\node[lbl, anchor=west] at (6.0*\sz, 1.7*\sz)
  {deflected path\\
   (solid blue):\\
   7 edges $= 5 + 2$};
\end{tikzpicture}
\caption{Deflection routing around a single failed node.  Under
XY-order deterministic routing on a 2-D grid, a packet whose
nominal shortest path crosses a dead node must detour one row
above (or below) the obstacle and rejoin the original corridor.
The additive cost of the detour is exactly $+2$ edges in the
isolated-obstacle case: one row-step and one column-step around
the obstacle.  Proposition~4 extends this observation to a
random-failure setting and controls the expected cumulative
detour across clusters of failed nodes.}
\label{fig:deflection-cartoon}
\end{figure}

\subsection{Long-range shortcuts}
A pure $\sqrt{P} \times \sqrt{P}$ grid graph has diameter
$\Theta(\sqrt{P})$; the coordination lower bound of
Proposition~1 is correspondingly $\Omega(\sqrt{P})$ cycles.
Many applications of peer-to-peer grid graphs augment the
nearest-neighbour graph with a sparse budget of long-range
shortcuts --- in the Watts--Strogatz / Newman--Watts tradition,
a small fraction of edges rewired or added as uniformly-random
partners.  The
typical shortest-path length then collapses to $O(\log P)$, and
the fault-tolerance analysis of Proposition~4 becomes
analytically tractable under the mean-field
(Erd\H{o}s--R\'enyi) universality assumption (rigorous for the
1-D-ring base, conjectural for the 2-D-grid base).  Figure~\ref{fig:smallworld-cartoon} overlays the shortcut
edges on the nearest-neighbour grid.  Proposition~5 formalises
the construction and its consequences.

\begin{figure}[!t]
\centering
\begin{tikzpicture}[
  meshnode/.style={rectangle, draw=cybiontblue!65, fill=cybiontblue!08,
    line width=0.4pt, minimum size=0.55cm, inner sep=0pt},
  origin/.style={rectangle, draw=cybiontblue!95!black,
    fill=cybiontblue!40, line width=0.9pt,
    minimum size=0.55cm, inner sep=0pt},
  meshedge/.style={draw=cybiontblue!40, line width=0.35pt},
  express/.style={->, >=Stealth, draw=orange!80!black, line width=0.85pt,
    bend left=22, opacity=0.85},
  lbl/.style={font={\scriptsize\linespread{1.18}\selectfont}, text=darkgray, align=left, inner sep=1pt}
]
\def\sz{0.60}
\foreach \i in {0,...,6}
  \foreach \j in {0,...,6} {
    \node[meshnode] (n\i\j) at (\i*\sz, \j*\sz) {};
  }
\node[origin] at (0,0) {};
% Grid edges (horizontal)
\foreach \j in {0,...,6}
  \foreach \i in {0,...,5} {
    \pgfmathtruncatemacro{\ii}{\i+1}
    \draw[meshedge] (n\i\j) -- (n\ii\j);
  }
% Grid edges (vertical)
\foreach \i in {0,...,6}
  \foreach \j in {0,...,5} {
    \pgfmathtruncatemacro{\jj}{\j+1}
    \draw[meshedge] (n\i\j) -- (n\i\jj);
  }
% Express links (long-range shortcuts)
\draw[express] (n05) to (n60);
\draw[express] (n13) to (n56);
\draw[express] (n30) to (n04);
\draw[express] (n66) to (n20);
\draw[express] (n22) to (n64);
\draw[express] (n52) to (n15);

\node[lbl, anchor=west] at (7*\sz + 0.25, 5.8*\sz)
  {long-range\\ shortcut edges\\ (orange)};
\node[lbl, anchor=west] at (7*\sz + 0.25, 3.4*\sz)
  {local grid graph edges\\ (blue, Manhattan)};
\node[lbl, anchor=west] at (7*\sz + 0.25, 0.6*\sz)
  {typical distance drops\\ from $\Theta(\sqrt{P})$\\ to $O(\log P)$};
\end{tikzpicture}
\caption{The small-world extension.  A sparse set of long-range
shortcuts (orange arcs) is overlaid on the nearest-neighbour
grid graph (blue grid).  Each node acquires $k \ge 1$ uniformly-random
long-range edges.  The typical shortest-path length collapses from
$\Theta(\sqrt{P})$ to $O(\log P)$, and percolation on the
augmented graph is conjectured to lie in the mean-field
universality class (rigorous for the 1-D-ring base
\cite{watts1998,newman2000smallworld}; conjectural for the
2-D-grid base).  Proposition~5 states the consequences.}
\label{fig:smallworld-cartoon}
\end{figure}
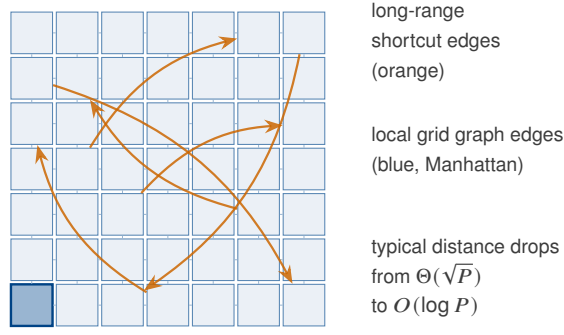

\subsection*{Ground rules for the proofs}

Every external theorem dependency is cited; elementary
derivations and modelling assumptions are stated explicitly.
Where a step initially required machinery I could not identify
in the cited bibliography, the audit process (see
\S\ref{sec:acks}) either supplied a closing derivation, a
further citation, a tightening of the proposition statement,
or a recognition that the step was not a mathematical claim.
I deliberately avoid confabulated theorem numbers: where I am
uncertain of an exact numbered result in a long monograph, I
cite the chapter or section heading.

\FloatBarrier

\section{Preliminaries}
\label{sec:prelim}

\subsection{Grid graph, cost, cycles}

Let $G = (V, E)$ be a finite connected graph embedded in
$\mathbb{Z}^2$ with edges joining nearest neighbours in the
Manhattan sense.  Write $|V| = P$ for the number of nodes and
$D := \operatorname{diam}(G)$ for the graph diameter.  The graph
metric $d_G$ coincides with the restricted Manhattan distance
$d_1(x, y) = |x_1 - y_1| + |x_2 - y_2|$.

A computation proceeds in synchronous \emph{cycles}.  In each cycle,
at most one payload may traverse each edge; each node may execute at
most one local update operation.  Write $t_{\mathrm{cycle}}$ for the
wall-clock duration of one cycle,
$t_{\mathrm{edge}} \in \mathbb{Z}_{\ge 0}$ for the cycle budget of a
single inter-node hop, and $t_{\mathrm{merge}}$ for the cycle budget
of one associative merge at a node.

\subsection{State measures and reduction target}

Assume a distributed computation has $P$ state partitions indexed
$i = 1, \dots, P$, with partition $i$ resident on node $x_i \in V$,
carrying positive mass $m_i > 0$.  Define the initial discrete
probability measure
\begin{equation}
  \mu = \sum_{i=1}^{P} a_i \, \delta_{x_i},
  \qquad
  a_i = \frac{m_i}{\sum_j m_j}.
\end{equation}
A global reduction converges the mass to a designated origin
$x_\star \in V$, giving the target
\begin{equation}
  \nu = \delta_{x_\star}.
\end{equation}
The admissible transport plans are
\(
  \Pi(\mu,\nu) = \{\gamma \in \mathcal{P}(V \times V):
    (\pi_1)_\# \gamma = \mu,\; (\pi_2)_\# \gamma = \nu\},
\)
and the 1-Wasserstein distance on $(V, d_1)$ is
\begin{equation}
  W_1(\mu,\nu) \;=\;
    \min_{\gamma \in \Pi(\mu,\nu)}
    \int_{V \times V} d_1(x,y) \, d\gamma(x,y).
\end{equation}
Existence of a minimising $\gamma$ is a finite-dimensional linear
programme and is an elementary consequence of the general existence
theorem for Kantorovich minimisers on Polish spaces
\cite{villani2009}.

\subsection{Admissibility: the abelian-monoid decomposition}

Let $S$ be a global state space and $p$ a traversing payload.  An
\emph{update rule}
\(
  U:\ \mathrm{State} \times \mathrm{Payload} \to \mathrm{State}
\)
is said to admit an \emph{abelian-monoid decomposition} if there
exist a localised map $g(p, S_i)$ and a binary operation $\oplus$
such that
\begin{equation}
  U(S, p) = \bigoplus_{i \in V} g(p, S_i),
\end{equation}
and $(A, \oplus, 0)$ is an abelian monoid on the carrier set
$A = \{ g(p, S_i) : i \in V, p \in \mathrm{Payload} \}$.

\subsection{Failure model}

Tile failures are modelled as independent site deletions on the
lattice $\mathbb{Z}^2$: each node fails with probability
$\delta \in [0, 1]$, independently of other nodes.  A \emph{failure
cluster} is a maximal connected subgraph of the induced
site-percolation configuration.  Following the convention of
\cite{grimmett1999}, write $p_c^{\mathrm{site}}(\mathbb{Z}^2)$
for the critical site-percolation threshold on the square
lattice.  Unlike the bond result
$p_c^{\mathrm{bond}}(\mathbb{Z}^2) = 1/2$ \cite{kesten1980}, no
closed-form value is known; the numerical estimate is
$p_c^{\mathrm{site}}(\mathbb{Z}^2) \approx 0.593$
\cite{akhunzhanov2022}.  We therefore work with the abstract
subcritical condition $\delta < p_c^{\mathrm{site}}(\mathbb{Z}^2)$
throughout.

\section{Proposition 1 --- Monge--Kantorovich lower bounds}
\label{sec:prop1}

\begin{proposition}\label{prop:p1}
Any synchronous schedule realising the reduction
$\mu \mapsto \nu$ on the grid graph $(V, d_1)$ is constrained by
three distinct lower bounds, each bounding a different quantity.

\textbf{(i) Transport-work bound.}  Total distance-weighted
payload movement is at least the $1$-Wasserstein distance,
\begin{equation}
  \sum_{i} a_i \, \ell_i \;\ge\; W_1(\mu, \nu),
\end{equation}
attained with equality by \emph{every} shortest-path schedule
(monotone Manhattan routing makes the inequality an identity for
each source independently; see Step~3 of the proof).

\textbf{(ii) Completion-depth bound.}  The minimum number of
synchronous rounds is at least the $\mu$-support radius,
\begin{equation}
  D_{\min}
    \;\ge\; r_\mu
    \;:=\; \max\bigl\{d_1(x, x_\star) : x \in \operatorname{supp}(\mu)\bigr\},
\end{equation}
attained by the parallel monotone Manhattan shortest-path
schedule (Step~3 of the proof) under the idealised non-congesting
model in which many shortest paths to $x_\star$ can share edges
freely.  Under edge-capacity-one
constraints the attainment is subject to a congestion-driven
scheduling term whose magnitude is recorded after the proof
(Lemma~\ref{lem:p1-variance} and the bottleneck-variance
paragraph).  The wall-clock completion time satisfies
\begin{equation}
  T_{\mathrm{grid}}
    \;\ge\; r_\mu \cdot t_{\mathrm{edge}} \cdot t_{\mathrm{cycle}}
\end{equation}
(hops $\times$ cycles-per-hop $\times$ seconds-per-cycle), and in
the idealised non-congesting regime a parallel shortest-path
schedule achieves $T_{\mathrm{grid}} \le (r_\mu t_{\mathrm{edge}}
+ K_{\mathrm{arch}}) t_{\mathrm{cycle}}$ for a fixed
architecture-dependent cycle overhead $K_{\mathrm{arch}}$
(Step~5 of the proof).

\textbf{(iii) Compressive-reduction edge bound.}  If the
computation is a fixed-size compressive reduction over a monoid
and each active source must influence the final sink value, the
set of edges carrying non-identity information during the
schedule must contain a connected subgraph spanning
$\operatorname{supp}(\mu) \cup \{x_\star\}$.  The number of
distinct used edges is therefore at least the graph-Steiner cost
\begin{equation}
  \operatorname{St}_G\bigl(\operatorname{supp}(\mu) \cup \{x_\star\}\bigr),
\end{equation}
where $\operatorname{St}_G(A)$ denotes the minimum number of edges in a
connected subgraph of $G$ containing $A$.

Bounds (i), (ii), (iii) are for three different quantities ---
transport work, causal depth, edges used --- and need not
saturate on the same schedule.  The weaker inequality
$D_{\min} \ge W_1(\mu, \nu)$ follows from $r_\mu \ge W_1(\mu,
\nu)$ and is tight only when the support of $\mu$ lies on a
single graph-distance sphere around $x_\star$.
\end{proposition}

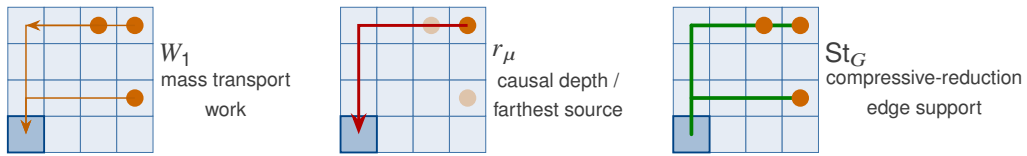
\begin{figure}[!b]
\centering
\begin{tikzpicture}[
  meshnode/.style={rectangle, draw=cybiontblue!75, fill=cybiontblue!10,
    line width=0.35pt, minimum size=0.48cm, inner sep=0pt},
  originsm/.style={rectangle, draw=cybiontblue!95!black,
    fill=cybiontblue!35, line width=0.7pt,
    minimum size=0.48cm, inner sep=0pt},
  srcdot/.style={circle, draw=none, fill=orange!80!black,
    minimum size=2.3mm, inner sep=0pt},
  transportflow/.style={->, >=Stealth, draw=orange!75!black, line width=0.55pt},
  depthflow/.style={->, >=Stealth, draw=red!70!black, line width=1.1pt},
  steineredge/.style={draw=green!50!black, line width=1.4pt, line cap=round},
  panellbl/.style={font=\small\bfseries, text=darkgray},
  lbl/.style={font={\scriptsize\linespread{1.18}\selectfont}, text=darkgray, align=center, inner sep=1pt}
]
\def\sz{0.48}
% --- Panel A: W_1 mass transport work ---
\begin{scope}[xshift=0cm]
  \foreach \i in {0,...,3}
    \foreach \j in {0,...,3}
      \node[meshnode] at (\i*\sz, \j*\sz) {};
  \node[originsm] at (0,0) {};
  \node[srcdot] at (3*\sz, 3*\sz) {};
  \node[srcdot] at (2*\sz, 3*\sz) {};
  \node[srcdot] at (3*\sz, 1*\sz) {};
  \draw[transportflow] (3*\sz,3*\sz)--(0,3*\sz)--(0,0);
  \draw[transportflow] (2*\sz,3*\sz)--(0,3*\sz);
  \draw[transportflow] (3*\sz,1*\sz)--(0,1*\sz)--(0,0);
  \node[panellbl, anchor=west] at (3*\sz+0.2, 2.2*\sz) {$W_1$};
  \node[lbl, anchor=west] at (3*\sz+0.3, 1.1*\sz) {mass transport\\ work};
\end{scope}
% --- Panel B: r_mu causal depth ---
\begin{scope}[xshift=4.4cm]
  \foreach \i in {0,...,3}
    \foreach \j in {0,...,3}
      \node[meshnode] at (\i*\sz, \j*\sz) {};
  \node[originsm] at (0,0) {};
  \node[srcdot] at (3*\sz, 3*\sz) {};
  \node[srcdot, opacity=0.30] at (2*\sz, 3*\sz) {};
  \node[srcdot, opacity=0.30] at (3*\sz, 1*\sz) {};
  \draw[depthflow] (3*\sz,3*\sz)--(0,3*\sz)--(0,0);
  \node[panellbl, anchor=west] at (3*\sz+0.2, 2.2*\sz) {$r_\mu$};
  \node[lbl, anchor=west] at (3*\sz+0.3, 1.1*\sz) {causal depth /\\ farthest source};
\end{scope}
% --- Panel C: St_G compressive-reduction edge support ---
\begin{scope}[xshift=8.8cm]
  \foreach \i in {0,...,3}
    \foreach \j in {0,...,3}
      \node[meshnode] at (\i*\sz, \j*\sz) {};
  \node[originsm] at (0,0) {};
  \draw[steineredge] (0,0)--(0,1*\sz)--(0,2*\sz)--(0,3*\sz);
  \draw[steineredge] (0,3*\sz)--(1*\sz,3*\sz)--(2*\sz,3*\sz)--(3*\sz,3*\sz);
  \draw[steineredge] (0,1*\sz)--(1*\sz,1*\sz)--(2*\sz,1*\sz)--(3*\sz,1*\sz);
  \node[srcdot] at (3*\sz, 3*\sz) {};
  \node[srcdot] at (2*\sz, 3*\sz) {};
  \node[srcdot] at (3*\sz, 1*\sz) {};
  \node[panellbl, anchor=west] at (3*\sz+0.2, 2.2*\sz) {$\operatorname{St}_G$};
  \node[lbl, anchor=west] at (3*\sz+0.3, 1.1*\sz) {compressive-reduction\\ edge support};
\end{scope}
\end{tikzpicture}
\caption{Three different quantities on the same grid instance.
Mass transport pays for every unit-distance moved ($W_1$);
causal depth is controlled by the farthest active source
($r_\mu$); a compressive monoid reduction only needs a
connected edge support spanning the sources and the sink
($\operatorname{St}_G$).  The three quantities need not
saturate on the same schedule.}
\label{fig:three-cost-cartoon}
\end{figure}

\begin{proof}
(Geometry: Figure~\ref{fig:grid-cartoon}; three lower-bound
quantities illustrated in Figure~\ref{fig:three-cost-cartoon}.)
\textbf{Step 1 --- Optimal coupling for a Dirac target.}
Because $\nu = \delta_{x_\star}$ places all mass at one point, the
marginal constraints force the coupling to be unique:
\(
  \gamma = \sum_{i=1}^{P} a_i \, \delta_{(x_i, x_\star)}.
\)
Consequently
\begin{equation}
  W_1(\mu, \nu)
    = \sum_{i=1}^{P} a_i \, d_1(x_i, x_\star),
\end{equation}
which is the mean Manhattan distance of the support of $\mu$ to the
origin.  This uses only the definition of $W_1$ and uniqueness of
Kantorovich couplings for a Dirac marginal, both elementary in the
finite-metric-space case \cite[Ch.~6]{villani2009}.

\textbf{Step 2 --- Causal depth: radius and Wasserstein bounds.}
In one synchronous cycle a payload traverses at most one grid
edge, so information initially at $x_i$ cannot influence the sink
$x_\star$ before round $d_1(x_i, x_\star)$.  Two lower bounds on
$D_{\min}$ (the minimum number of synchronous rounds required
for the reduction to complete) follow.

\emph{Strong (radius) bound.}  Because $D_{\min} \ge d_1(x_i,
x_\star)$ holds for every $i \in \operatorname{supp}(\mu)$, the
completion depth is at least the maximum,
\begin{equation}
  D_{\min} \;\ge\; r_\mu
    \;:=\; \max\bigl\{d_1(x, x_\star) : x \in \operatorname{supp}(\mu)\bigr\},
\end{equation}
which is bound (ii) of the proposition.

\emph{Weak (Wasserstein) bound.}  Multiplying each individual
inequality $D_{\min} \ge d_1(x_i, x_\star)$ by $a_i \ge 0$ and
summing, using $\sum_i a_i = 1$,
\begin{equation}
  D_{\min}
    \;=\; D_{\min} \sum_i a_i
    \;\ge\; \sum_i a_i\, d_1(x_i, x_\star)
    \;=\; W_1(\mu, \nu),
\end{equation}
with equality iff all sources sit at a single distance from
$x_\star$; this is the weaker consequence noted in the
proposition.

Bound (i) concerns transport work $\sum_i a_i \ell_i$ rather
than depth: any path from $x_i$ to $x_\star$ has length
$\ell_i \ge d_1(x_i, x_\star)$, so
$\sum_i a_i \ell_i \ge W_1(\mu, \nu)$.  The statements are
compatible with the time-expanded-network formalism for flows
over time~\cite{skutella2009}, but the elementary arguments
above do not require it.

\textbf{Step 3 --- Attainment by shortest-path routing.}
Monotone Manhattan paths give $\ell_i = d_1(x_i, x_\star)$ for
every source.  Two attainment statements follow, each for a
different quantity.

\emph{Work attainment.}  The transport cost matches the
Wasserstein distance exactly,
\begin{equation}
  \sum_{i=1}^{P} a_i \, \ell_i
    \;=\; \sum_{i=1}^{P} a_i \, d_1(x_i, x_\star)
    \;=\; W_1(\mu, \nu),
\end{equation}
so bound (i) is tight for any shortest-path schedule, serial or
parallel.

\emph{Depth attainment.}  Under non-congesting parallel
dispatch --- each source sends its payload simultaneously along
its own monotone dimension-order Manhattan path --- the schedule
completes in exactly $\max_i d_1(x_i, x_\star) = r_\mu$ rounds,
so bound (ii) is tight.  With bandwidth limits the completion
time degenerates toward the total-work figure, approaching
$\sum_i \ell_i$ rather than $r_\mu$ rounds.

\textbf{Step 4 --- Continuous elevation: Brenier under
regularisation.}
To obtain the measure-theoretic elevation, embed the finite grid
into $\Omega \subset \mathbb{R}^2$ by identifying each node
$x \in V$ with a unit cell $Q_x \subset \Omega$, and define
mollified absolutely-continuous approximants
\begin{equation}
  \mu_\varepsilon
    = \sum_{i=1}^{P} a_i \, |Q_i^\varepsilon|^{-1} \,
      \mathbf{1}_{Q_i^\varepsilon}(x) \, dx,
  \qquad
  \nu_\varepsilon
    = |Q_\star^\varepsilon|^{-1} \,
      \mathbf{1}_{Q_\star^\varepsilon}(x) \, dx,
\end{equation}
where $Q_i^\varepsilon, Q_\star^\varepsilon$ are shrinking squares
centred at $x_i, x_\star$.  For the quadratic cost
$c(x,y) = |x-y|^2$, Brenier's polar factorisation theorem
\cite{brenier1991} gives a unique optimal transport map
$T_\varepsilon = \nabla \phi_\varepsilon$ as the gradient of a
convex potential, provided $\mu_\varepsilon$ is absolutely
continuous with respect to Lebesgue measure
\cite{figalli2014,villani2009}.  Passing from the Euclidean $W_2$
Brenier picture back to the discrete $W_1$ problem on the grid
graph is not a direct application of Brenier's theorem --- the
source in the hardware problem is atomic and the cost is
Manhattan.  The continuous Brenier map $T_\varepsilon$ does
not constrain the discrete XY routing schedule operationally:
it is a conceptual elevation of the regularised problem
$(\mu_\varepsilon, \nu_\varepsilon)$, not a bound on the
discrete routing algorithm.  Peyr\'e and
Cuturi~\cite[§3.1]{peyre2019} make the point sharply: for
discrete measures on finite metric spaces, Monge maps need not
exist, and Brenier's theorem does not apply literally; the
correct formulation at the discrete level remains the
Kantorovich relaxation, which allows mass-splitting through
couplings.  The discrete $W_1$ lower bound of Step~2 is the
mathematically binding statement; the Brenier picture merely
records that this lower bound sits inside a well-understood
continuous optimal-transport framework.

\textbf{Step 5 --- Wall-clock bound.}
Let $K_{\mathrm{arch}} \ge 0$ absorb inter-node forwarding and
local-merge overhead as a fixed architecture-dependent cycle
count.  Combining Steps~2--4, the wall-clock completion time
(hops $\times$ cycles-per-hop $\times$ seconds-per-cycle)
satisfies
\begin{equation}
  T_{\mathrm{grid}} \;\ge\;
    r_\mu \cdot t_{\mathrm{edge}} \cdot t_{\mathrm{cycle}},
\end{equation}
and, in the idealised non-congesting regime, a parallel
shortest-path schedule achieves
\begin{equation}
  T_{\mathrm{grid}}
    \;\le\;
    \bigl(r_\mu \cdot t_{\mathrm{edge}}
      + K_{\mathrm{arch}}\bigr) \cdot t_{\mathrm{cycle}}.
\end{equation}
The weaker inequality $T_{\mathrm{grid}} \ge W_1(\mu, \nu)
\cdot t_{\mathrm{edge}} \cdot t_{\mathrm{cycle}}$ follows from
$r_\mu \ge W_1(\mu, \nu)$.  The constant $K_{\mathrm{arch}}$ is
a finite non-negative integer fixed by graph realisation and is
not a mathematical object; all statements hold uniformly in
$K_{\mathrm{arch}}$.

\textbf{Step 5\textsuperscript{$\prime$} --- Compressive-reduction
edge bound.}  For bound (iii), suppose the computation is a
compressive monoid reduction and every active source must
influence the final sink value.  Let $E' \subseteq E$ be the
set of edges that carry non-identity information during the
execution.  If some $x \in \operatorname{supp}(\mu)$ were not
connected to $x_\star$ in the subgraph $(V, E')$, no
information from $x$ could reach $x_\star$, contradicting the
assumption that $x$ influences the final aggregate.  Hence
$(V, E')$ contains a connected subgraph spanning
$\operatorname{supp}(\mu) \cup \{x_\star\}$, so
$|E'| \ge \operatorname{St}_G(\operatorname{supp}(\mu) \cup \{x_\star\})$,
the graph-Steiner-tree cost.

\textbf{Step 6 --- Maximum edge congestion under all-to-one
routing.}  The route-incidence matrix $M \in \{0,1\}^{P \times
|E|}$ has rows indexed by sources (each with one deterministic
route to $x_\star$) and columns indexed by edges.  Under
deterministic dimension-order routing with a single sink, the
maximum edge congestion
$N_{\max} := \max_e |\{r : M_{r,e} = 1\}|$
is attained on the sink-adjacent edges of the sink trunk and is
$\Theta(P)$, not $\Theta(\sqrt{P})$ as a uniform-sink heuristic
would suggest.  For the corner-sink variant
$x_\star = (0, 0)$ on the $L \times L$ grid, the edge between
$(0, 1)$ and $(0, 0)$ carries all $L \cdot (L-1) = \Theta(P)$
sources above row zero, as the computation in
Lemma~\ref{lem:p1-variance} makes explicit.  Attainment
of bound (ii) in the non-congesting idealisation therefore
requires either a routing scheme that balances load away from
the sink trunk (Valiant-style two-phase routing), compressive
aggregation that replaces the trunk bottleneck by a balanced
reduction tree (bound (iii)), or shortcut augmentation
(\S\ref{sec:prop5}).  This sets the stage for the negative
result below.

\end{proof}

\paragraph{Bottleneck variance: a negative result for
corner-sink dimension-order routing.}
A plausible follow-up to Proposition~\ref{prop:p1} is a
concentration claim for the completion time of an individual
route under independent Bernoulli source activation.  The
following proposition shows that the naive expectation ``edge
loads are weakly correlated, so per-route completion time
concentrates as in \S6 of \cite{tropp2012}'' fails in the
simplest deterministic setting: the variance of the
completion-time functional along the sink-column trunk scales
as $\Theta(f_{\mathrm{act}}(1-f_{\mathrm{act}})P^2)$, dominated
by a single bottleneck, rather than as any $o(P^2)$ quantity.

\begin{lemma}[Bottleneck variance for corner-sink
dimension-order routing]\label{lem:p1-variance}
Let $G_n = \{0, \dots, n-1\}^2$ be the $n \times n$ grid with
$P = n^2$, let $x_\star = (0, 0)$, and route each source
$u = (a, b)$ first horizontally to $(0, b)$ and then vertically
along the sink column to $(0, 0)$.  Let $X_u \sim
\mathrm{Bernoulli}(f_{\mathrm{act}})$ be independent activation
indicators with $f_{\mathrm{act}} \in (0, 1)$, define the
directed edge load $\Lambda_e := \sum_{u : e \in r(u)} X_u$,
and let $r_v$ be the vertical route from the top-column source
$v = (0, n-1)$ to the sink.  Write the aggregate route-load
functional
\begin{equation}
  Y_v \;:=\; \sum_{e \in r_v} \Lambda_e
\end{equation}
(a load count, not directly a wall-clock time).  Then
\begin{equation}
  \operatorname{Var}(Y_v)
    \;=\; \Theta\bigl(f_{\mathrm{act}}(1 - f_{\mathrm{act}}) n^4\bigr)
    \;=\; \Theta\bigl(f_{\mathrm{act}}(1 - f_{\mathrm{act}}) P^2\bigr).
\end{equation}
In particular, any uniform upper bound of the form
$\operatorname{Var}(Y_v) = O(f_{\mathrm{act}} P^{3/2})$ is false
for this routing model.
\end{lemma}

\begin{proof}
(Sink-trunk geometry and load-weighting:
Figure~\ref{fig:bottleneck-cartoon}.)
Index the sink-column edges by
$e_j := ((0, j), (0, j-1))$ for $j = 1, \dots, n-1$; the route
$r_v$ consists of exactly these $n-1$ edges.  A source $u = (a,
b)$ uses edge $e_j$ iff $b \ge j$, so
$\Lambda_{e_j} = \sum_{a = 0}^{n-1} \sum_{b = j}^{n-1}
X_{a, b}$.  Interchanging the summation order in
$Y_v = \sum_{j = 1}^{n-1} \Lambda_{e_j}$ gives
\begin{equation}
  Y_v
    \;=\; \sum_{a = 0}^{n-1} \sum_{b = 1}^{n-1} b \, X_{a, b},
\end{equation}
since a source in row $b$ contributes to exactly the vertical
edges $e_1, \dots, e_b$ and therefore has weight $b$.  The
$X_{a, b}$ are independent Bernoulli random variables with
variance $f_{\mathrm{act}}(1 - f_{\mathrm{act}})$, so
\begin{equation}
  \operatorname{Var}(Y_v)
    \;=\; f_{\mathrm{act}}(1 - f_{\mathrm{act}})
          \sum_{a = 0}^{n-1} \sum_{b = 1}^{n-1} b^{2}
    \;=\; f_{\mathrm{act}}(1 - f_{\mathrm{act}}) \, n
          \cdot \frac{(n-1) n (2n-1)}{6}
    \;=\; \Theta\bigl(f_{\mathrm{act}}(1 - f_{\mathrm{act}}) n^4\bigr).
\end{equation}
Substituting $P = n^2$ gives the stated order.  A
$O(f_{\mathrm{act}} P^{3/2})$ bound is smaller by a factor
$\Theta(\sqrt{P})$ and is therefore refuted.
\end{proof}

\begin{figure}[!b]
\centering
\begin{tikzpicture}[
  meshnode/.style={rectangle, draw=cybiontblue!75, fill=cybiontblue!10,
    line width=0.35pt, minimum size=0.44cm, inner sep=0pt},
  originsm/.style={rectangle, draw=cybiontblue!95!black,
    fill=cybiontblue!35, line width=0.7pt,
    minimum size=0.44cm, inner sep=0pt},
  srcdot/.style={circle, draw=none, fill=orange!80!black,
    minimum size=1.6mm, inner sep=0pt},
  trunkedge/.style={draw=red!75!black, line cap=round},
  horizleg/.style={->, >=Stealth, draw=orange!70!black, line width=0.35pt,
    dash pattern=on 1pt off 1.3pt, opacity=0.55},
  lbl/.style={font={\scriptsize\linespread{1.18}\selectfont}, text=darkgray, align=left, inner sep=1pt}
]
\def\sz{0.44}
\def\LL{5}
% grid
\foreach \i in {0,...,\LL}
  \foreach \j in {0,...,\LL}
    \node[meshnode] at (\i*\sz, \j*\sz) {};
\node[originsm] at (0,0) {};
% sources (all non-sink nodes active, depicted as small dots)
\foreach \i in {1,...,\LL}
  \foreach \j in {0,...,\LL}
    \node[srcdot] at (\i*\sz, \j*\sz) {};
\foreach \j in {1,...,\LL}
  \node[srcdot] at (0, \j*\sz) {};
% horizontal legs to column 0 (faint)
\foreach \j in {1,...,\LL}
  \draw[horizleg] (\LL*\sz, \j*\sz) -- (0, \j*\sz);
% trunk edges: thickness scales with load ~ n(n-j+1)
\draw[trunkedge, line width=3.1pt] (0, 0) -- (0, 1*\sz);
\draw[trunkedge, line width=2.5pt] (0, 1*\sz) -- (0, 2*\sz);
\draw[trunkedge, line width=2.0pt] (0, 2*\sz) -- (0, 3*\sz);
\draw[trunkedge, line width=1.5pt] (0, 3*\sz) -- (0, 4*\sz);
\draw[trunkedge, line width=1.0pt] (0, 4*\sz) -- (0, 5*\sz);
% annotations
\node[lbl, anchor=west] at (\LL*\sz+0.25, 0.5*\sz)
  {$\Lambda_{e_1}\!=\!\Theta(P)$\\ sink-adjacent};
\node[lbl, anchor=west] at (\LL*\sz+0.25, 4.5*\sz)
  {$\Lambda_{e_{n-1}}\!=\!\Theta(\sqrt{P})$\\ top of trunk};
\node[lbl, anchor=west] at (\LL*\sz+0.25, 2.5*\sz)
  {trunk edge thickness\\ $\propto$ load $\Lambda_{e_j}$};
\node[lbl, anchor=north] at (0, -0.25)
  {sink $x_\star$};
\node[lbl, anchor=east] at (-0.2, 3*\sz)
  {column~0\\ (sink column)};
\end{tikzpicture}
\caption{Why deterministic corner-sink dimension-order routing
creates coherent variance.  Every source routes horizontally
to column~$0$ (faint orange arrows), then vertically down the
sink column to $x_\star = (0,0)$; sink-column trunk edges near
the sink therefore carry $\Theta(P)$ sources (thickest red
segments).  A source in row $b$ contributes to $b$ trunk
edges, so by Fubini $Y_v = \sum_{a,b} b \, X_{a,b}$ and
$\operatorname{Var}(Y_v) = \Theta(f_{\mathrm{act}}(1 -
f_{\mathrm{act}}) P^2)$.}
\label{fig:bottleneck-cartoon}
\end{figure}
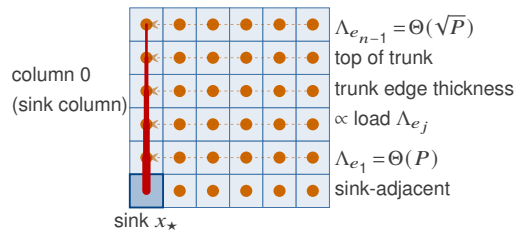

\paragraph{Consequence for concentration claims.}
The bottleneck arises because the sink-column edges have
upstream sets of size $n(n - j + 1) = \Theta(P)$ for small
$j$, and the same source contributes coherently to all
vertical edges between its row and the sink.  Under serial
edge service, $Y_v$ lower-bounds the congestion-induced
completion time along the trunk; a route-load variance bound
of order $O(f_{\mathrm{act}} P^{3/2})$ would therefore require
an additional upstream-load assumption
$\max_{e \in r} |\{u : e \in r(u)\}| = O(\sqrt{P})$ which is
violated by corner-sink dimension-order routing.  The
expectation along the trunk is $\mathbb{E}[Y_v] = \Theta(
f_{\mathrm{act}} n^3) = \Theta(f_{\mathrm{act}} P^{3/2})$, so
relative concentration
$\operatorname{Var}(Y_v) / \mathbb{E}[Y_v]^2 = \Theta((1 -
f_{\mathrm{act}}) / (f_{\mathrm{act}} P))$ does still hold for
this functional, but the deterministic routing scheme exposes
the bottleneck in absolute variance rather than eliminating
it.  The architectural
consequence is that a grid substrate committed to
single-sink-saturation by deterministic routing should either
use compressive aggregation (bound (iii) of
Proposition~\ref{prop:p1}, which constrains edges-used rather
than per-route variance), load-balanced routing (e.g.\
Valiant-style two-phase randomisation), or shortcut
augmentation (\S\ref{sec:prop5}).

\section{Proposition 2 --- Sparse-participation scaling advantage}
\label{sec:prop2}

\begin{proposition}\label{prop:p2}
Let $f_{\mathrm{act}} \in (0, 1]$ be the fraction of participants
active per output, let $P$ be the grid graph size, and let $N$ be
the number of participants in a dense communication graph.  Under
the latency decompositions stated below --- a geometry-dominated
grid-side model and a fixed-overhead dense-collective model,
both explicit assumptions about communication-fabric behaviour
rather than graph-theoretic facts --- the ratio
$T_{\mathrm{cluster}}(f_{\mathrm{act}}, N) /
 T_{\mathrm{grid}}(f_{\mathrm{act}}, P)$ is strictly monotonically
increasing in $1/f_{\mathrm{act}}$ once a single inequality
relating the communication-overhead constants of the two fabrics
holds.  Literal divergence requires an additional asymptotic in
$N$ or in the ratio $M_P / (A_N + B_N)$.
\end{proposition}

\begin{proof}
(Sparse-activation pattern: Figure~\ref{fig:sparse-cartoon}.)
\textbf{Step 1 --- Grid graph-side latency model.}
For a grid-fold reduction on a 2-D grid of $P$ nodes, the fold
depth is constrained by graph geometry.  On a square grid graph the
diameter is $\Theta(\sqrt{P})$; a balanced tree-fold over
independent participants has depth $\Theta(\log P)$
\cite{blelloch1990}.  The binding constraint is geometric, so the
composite model is
\begin{equation}
  T_{\mathrm{grid}}(f_{\mathrm{act}}, P)
    \;=\; c_1 f_{\mathrm{act}}
      + c_w \sqrt{P} \cdot t_{\mathrm{edge}}
      + O(\log P) \cdot t_{\mathrm{merge}}
    \;=\; c_1 f_{\mathrm{act}} + M_P,
\end{equation}
where $M_P$ absorbs the $\sqrt{P}$-dominated routing depth.
The justification for treating the fold as the product-preserving
image of a commutative-monoid reduction is the content of
Proposition~\ref{prop:p3}.

\textbf{Step 2 --- Cluster-side latency model.}
Write $\alpha$ for the startup latency of a single collective hop,
$\beta = 1/B$ for the per-byte bandwidth cost, and $\gamma$ for
the per-byte reduction cost.
Thakur, Rabenseifner, and Gropp~\cite[\S3]{thakur2005} give
explicit $\alpha$-$\beta$-$\gamma$ formulas for several collective
variants on $p$ participants with message size $n$ bytes.  For
recursive doubling at power-of-two $p$,
\begin{equation}
  T_{\mathrm{rec\text{-}dbl}}(p, n)
    = \log p \cdot \alpha
      + n \log p \cdot \beta
      + n \log p \cdot \gamma,
\end{equation}
and for the Rabenseifner long-message all-reduce,
\begin{equation}
  T_{\mathrm{Rab}}(p, n)
    = 2 \log p \cdot \alpha
      + 2 \, \frac{p - 1}{p} \cdot n \beta
      + \frac{p - 1}{p} \cdot n \gamma.
\end{equation}
Patarasuk and Yuan~\cite{patarasuk2009} give the bandwidth-optimal
ring all-reduce at cost $2 \frac{p-1}{p} \cdot n \beta + 2(p-1)
\cdot \alpha$ and note that it is not latency-optimal, whereas
butterfly/tree variants are latency-optimal but not
bandwidth-optimal.  Writing out the composite form in the
sparse-regime variables below, and collecting the
$f_{\mathrm{act}}$-dependent and -independent contributions,
\begin{equation}
  T_{\mathrm{cluster}}(f_{\mathrm{act}}, N)
    \;=\; c_2 f_{\mathrm{act}}
      \;+\; \Theta\!\left(\tfrac{N - 1}{N} \cdot
            \tfrac{m(f_{\mathrm{act}})}{B}\right)
      \;+\; \Theta(\log N) \cdot \alpha.
\end{equation}
This is the only latency form the monotonicity argument below
requires.

\textbf{Step 3 --- Sparse-regime behaviour.}
For small $f_{\mathrm{act}}$, assume the communicated message size
is lower-bounded by a regime-specific floor,
$m(f_{\mathrm{act}}) \ge m_0 > 0$; this captures either fixed
control-plane traffic or a non-vanishing per-output synchronisation
volume.  Under that condition,
\begin{equation}
  T_{\mathrm{cluster}}(f_{\mathrm{act}}, N)
    \;\ge\; c_2 f_{\mathrm{act}} + A_N + B_N,
\end{equation}
with $A_N$ and $B_N$ denoting the $f_{\mathrm{act}}$-independent
bandwidth and latency overheads respectively.

\textbf{Step 4 --- Workload-output invariance under
abelian-monoid fold.}
For the comparison ratio $T_{\mathrm{cluster}} /
T_{\mathrm{grid}}$ to be meaningful, the two architectures must
agree on the output of the \emph{same} workload rather than
merely run structurally analogous schedules.  For the
workload class covered by Proposition~\ref{prop:p3} ---
reductions whose update rule factors into a local map and an
abelian-monoid merge --- the output is a pure function of the
multiset of local values $\{g(p, S_i)\}_{i \in V}$: by
commutativity and associativity of $\oplus$, any reduction
schedule on any graph topology returns the same element of $A$.
The two architectures therefore compute the same workload by
virtue of the algebra of the fold, not by any convergence of
learning dynamics or kernel operators.

The specific form $c_1 f_{\mathrm{act}}$ of the local compute
term is a \emph{sparse-participation cost model} justified at
the application layer.  In a graph organised as $E$
computational units per group where each output activates $k$
units, $f_{\mathrm{act}} = k/E$; if a dense evaluation incurs
per-unit cost $c_{\mathrm{unit}}$, the sparse evaluation
incurs $k \cdot c_{\mathrm{unit}} = f_{\mathrm{act}} \cdot
E \cdot c_{\mathrm{unit}}$, proportional to $f_{\mathrm{act}}$.
This sparse-cost pattern is the mechanism underlying several
production sparse-routing architectures~\cite{shazeer2017,
fedus2022switch}.

\textbf{Step 5 --- Monotonicity in $1/f_{\mathrm{act}}$.}
Write $x := 1/f_{\mathrm{act}}$ and form the ratio
\begin{equation}
  R(x)
    \;=\; \frac{T_{\mathrm{cluster}}}{T_{\mathrm{grid}}}
    \;=\; \frac{c_2 / x + A_N + B_N}{c_1 / x + M_P}
    \;=\; \frac{c_2 + (A_N + B_N) x}{c_1 + M_P \cdot x}.
\end{equation}
Differentiating,
\begin{equation}
  R'(x) \;=\;
    \frac{(A_N + B_N) c_1 - M_P c_2}{(c_1 + M_P x)^2}.
\end{equation}
Thus $R'(x) > 0$ if and only if
\begin{equation}
  (A_N + B_N) c_1 \;>\; M_P c_2.
\end{equation}
This is precisely the condition that the
$f_{\mathrm{act}}$-independent cluster overhead, measured against
the grid routing constant, dominates the corresponding grid graph
scaling term.  Under that inequality, $R$ is strictly
increasing in $x = 1/f_{\mathrm{act}}$; equivalently, the ratio
improves strictly in favour of the grid graph as sparsity grows.
Note that $R(x)$ uses the Step~3 lower bound on $T_{\mathrm{cluster}}$;
the omitted terms in $T_{\mathrm{cluster}}$ (the higher-order
contributions of the Step~2 expansion) are
$f_{\mathrm{act}}$-independent constants under the stated latency
model, so monotonicity of $R$ in $x$ transfers from the lower bound
to the actual ratio $T_{\mathrm{cluster}} / T_{\mathrm{grid}}$.

\textbf{Step 6 --- Strict monotone separation, not divergence.}
Taking $f_{\mathrm{act}} \to 0$ with $P, N$ fixed,
\begin{equation}
  \lim_{f_{\mathrm{act}} \to 0}
    \frac{T_{\mathrm{cluster}}}{T_{\mathrm{grid}}}
    \;=\; \frac{A_N + B_N}{M_P},
\end{equation}
which is a positive constant.  For literal divergence one must
supplement the model with a joint asymptotic, e.g.\ $N \to \infty$
with $P$ fixed, or $M_P = o(A_N + B_N)$.  Without such
supplementation the rigorous conclusion is strict monotone
separation, not divergence --- as recorded in the proposition
statement.
\end{proof}

\begin{corollary}[Divergence under unbounded participant scaling]
\label{cor:p2-divergence}
Under the latency decompositions of Steps~1--2 and the
sparse-regime assumption of Step~3 of the proof of
Proposition~\ref{prop:p2}, with $f_{\mathrm{act}} \in (0, 1]$ and
$P$ both fixed,
\begin{equation}
  \lim_{N \to \infty}
    \frac{T_{\mathrm{cluster}}(f_{\mathrm{act}}, N)}
         {T_{\mathrm{grid}}(f_{\mathrm{act}}, P)}
    \;=\; \infty,
\end{equation}
with explicit asymptotic rate
\begin{equation}
  \frac{T_{\mathrm{cluster}}(f_{\mathrm{act}}, N)}
       {T_{\mathrm{grid}}(f_{\mathrm{act}}, P)}
    \;=\; \Theta(\log N)
  \qquad (N \to \infty).
\end{equation}
\end{corollary}

\begin{proof}
From Step~2 of the proof of Proposition~\ref{prop:p2},
\begin{equation}
  T_{\mathrm{cluster}}(f_{\mathrm{act}}, N)
    \;=\; c_2 f_{\mathrm{act}}
    \,+\, \Theta\!\left(\tfrac{N-1}{N}\cdot\tfrac{m(f_{\mathrm{act}})}{B}\right)
    \,+\, \Theta(\log N) \cdot \alpha.
\end{equation}
With $f_{\mathrm{act}} \in (0, 1]$ and $P$ fixed and with
$m(f_{\mathrm{act}}) \ge m_0 > 0$ (Step~3), the first term is
$O(1)$ by the bound $f_{\mathrm{act}} \le 1$; the second term
is $\Theta(1)$ as $(N-1)/N \to 1$ and $m(f_{\mathrm{act}})/B$
is a fixed positive constant; the third term is
$\Theta(\log N) \cdot \alpha = \Theta(\log N)$ as $N \to \infty$
(with $\alpha > 0$ a fixed cluster-side startup-latency constant
inherited from the $\alpha$-$\beta$-$\gamma$ model of Step~2).
Summing the three asymptotic orders,
\begin{equation}
  T_{\mathrm{cluster}}(f_{\mathrm{act}}, N)
    \;=\; O(1) + \Theta(1) + \Theta(\log N)
    \;=\; \Theta(\log N)
  \qquad (N \to \infty).
\end{equation}

The grid-side cost
$T_{\mathrm{grid}}(f_{\mathrm{act}}, P) = c_1 f_{\mathrm{act}} + M_P$
depends only on $f_{\mathrm{act}}$ and $P$, both fixed, so
$T_{\mathrm{grid}}(f_{\mathrm{act}}, P) = \Theta(1)$ as
$N \to \infty$.

Taking the ratio,
\begin{equation}
  \frac{T_{\mathrm{cluster}}(f_{\mathrm{act}}, N)}
       {T_{\mathrm{grid}}(f_{\mathrm{act}}, P)}
    \;=\; \frac{\Theta(\log N)}{\Theta(1)}
    \;=\; \Theta(\log N)
    \;\to\; \infty
  \qquad (N \to \infty),
\end{equation}
which establishes both the unboundedness limit and the explicit
asymptotic rate.
\end{proof}

\begin{remark}
The corollary specialises the first of the two divergence
asymptotics flagged in Step~6 ($N \to \infty$ with $P$ fixed) to
an explicit $\Theta(\log N)$ rate.  The dual asymptotic
$M_P = o(A_N + B_N)$ with $N$ fixed and $P \to \infty$ yields
the analogous statement with rate $\Theta\bigl((A_N + B_N)/M_P\bigr)$,
which diverges whenever the antecedent holds; the proof is
structurally identical, replacing the $\log N$ growth of
$B_N$ by the assumed sub-dominance of $M_P$.  The two
asymptotics are not mutually exclusive: joint scaling
$N \to \infty$ and $P \to \infty$ at independent rates yields
intermediate divergence rates determined by which of
$B_N = \Theta(\log N)$ and $M_P = \Theta(\sqrt{P})$ dominates.
\end{remark}

\section{Proposition 3 --- Functorial admissibility and
proof-carrying code}
\label{sec:prop3}

\begin{proposition}\label{prop:p3}
A workload $W = (S, p, U)$ with state space $S$, payload $p$,
and update rule $U$ has \emph{schedule-independent reduction
semantics} if $U$ decomposes into a local map $g(p, S_i)$ and a
commutative associative merge $\oplus$ with $(A, \oplus, 0)$ an
abelian monoid: every binary-tree reduction of the same
multiset of inputs produces the same output.  Under that
decomposition, and additionally under an idealised
non-interfering scheduler --- the scheduling assumption used
to define $\mathbf{HwState}$ below, in which independent local
updates, inter-node hops, and merges proceed without competing
for shared physical resources in the same synchronous cycle
--- the wall-clock bound
\begin{equation}
  T_{\mathrm{wall}}
    \;\le\;
    \Bigl(\operatorname{diam}(G) \cdot
      (t_{\mathrm{edge}} + t_{\mathrm{merge}})
      + \max_i t_{\mathrm{local}}(i)\Bigr) \cdot t_{\mathrm{cycle}}
\end{equation}
(with $t_{\mathrm{edge}}, t_{\mathrm{merge}}, t_{\mathrm{local}}$
in cycles and $t_{\mathrm{cycle}}$ in seconds per cycle) holds,
with admissibility encodable as proof-carrying code via the
monoid-law proof terms of Step~6.
\end{proposition}

\begin{proof}
(Monoid-fold diagram: Figure~\ref{fig:fold-cartoon}.)
\textbf{Step 1 --- The category $\mathbf{HwState}$.}
Define a small category whose objects are global hardware
configurations $X = (s_v)_{v \in V}$, assigning a local state $s_v$
to each node $v \in V$, and whose morphisms are finite composable
sequences of allowed architectural transitions: local updates,
inter-node edges, and merges.  Identity morphisms are empty
transition sequences; composition is concatenation.  We restrict
the morphism class to \emph{non-interfering} transitions: those
that can be scheduled concurrently without competing for the
same physical edge or node in the same synchronous cycle.
Within this non-interfering subcategory, binary products of
independent node configurations exist componentwise and satisfy
the universal property, giving $\mathbf{HwState}$ the structure
of a category with finite products
\cite[Ch.~2]{maclane1971}.  Transitions that would compete for
shared resources are handled at scheduling time
(\S\ref{sec:prop4} routes around them) and are not part of the
categorical structure used here.

\textbf{Step 2 --- The Lawvere theory of commutative monoids.}
The Lawvere theory $\mathrm{Th}(\mathrm{CM})$ of commutative
monoids is the small category with finite products generated by a
single object together with morphisms realising a nullary unit
$0$ and a binary multiplication $\oplus$ subject to associativity,
commutativity, and unit laws.  A finite-product-preserving functor
$M : \mathrm{Th}(\mathrm{CM}) \to \mathbf{Set}$ is, by definition,
a commutative monoid; more generally, product-preserving functors
into any cartesian category classify monoid \emph{instances} in
that category \cite{lawvere1963,maclane1971,hylandpower2007}.

\textbf{Step 3 --- The admissibility functor $F$.}
Suppose the update rule decomposes as
\(
  U(S, p) = \bigoplus_{i \in V} g(p, S_i),
\)
with $(A, \oplus, 0)$ an abelian monoid.  Define a functor
\begin{equation}
  F : \mathrm{Th}(\mathrm{CM}) \longrightarrow \mathbf{HwState}
\end{equation}
by sending the distinguished generator to the object
``one node carrying a value in $A$'', sending finite products to
$n$-tuples of such nodes, and sending the theory's unit and
multiplication to hardware initialisation and hardware merge
respectively.  Because hardware merge physically implements
$\oplus$ and local preparation implements $g(p, S_i)$, the image
under $F$ of a formal algebraic term is the corresponding hardware
execution.  Preservation of finite products is equivalent to the
global update rule arising from local maps plus monoidal merging,
i.e.\ to the decomposition $(g, \oplus)$.  Existence of such
product-preserving functors is the content of Lawvere's functorial
semantics \cite[\S2]{lawvere1963}.  Uniqueness up to natural
isomorphism follows from generation by a single object:
$\mathrm{Th}(\mathrm{CM})$ is the Lawvere theory generated under
finite products by one distinguished generator $1$, so any
product-preserving functor $F$ is determined by the image $F(1)$
together with the interpretations of the unit
$0 : 1^{0} \to 1$ and the multiplication
$\oplus : 1 \times 1 \to 1$.  If two such functors $F, G$ encode
the same hardware semantics --- meaning an isomorphism
$\eta_1 : F(1) \cong G(1)$ that intertwines their interpretations
of $0$ and $\oplus$ --- the assignment
$\eta_n := \eta_1^{\times n} : F(n) \to G(n)$ extends to every
object.  Naturality for every morphism of
$\mathrm{Th}(\mathrm{CM})$ follows because each such morphism is
generated from finite products, projections, $0$, and $\oplus$,
all of which are preserved by $\eta_1$
\cite[Ch.~6]{maclane1971}.  Hence $F \cong G$ as functors.

\textbf{Step 4 --- Evaluation-order independence.}
For any finite multiset $\{a_1, \dots, a_n\} \subset A$, the
abelian-monoid laws make
\(
  a_1 \oplus a_2 \oplus \cdots \oplus a_n
\)
independent of association and of permutation.  All syntactic
reduction trees over the same multiset define the same morphism
in $\mathrm{Th}(\mathrm{CM})$; under $F$, all corresponding
hardware executions compute the same value.  A compiler is
therefore free to choose any wavefront tree-fold consistent with
the communication graph, at no change in semantics
\cite[\S2]{lawvere1963}.

\textbf{Step 5 --- Wall-clock bound by tree reduction.}
Choose an origin $r \in V$ and reduce toward $r$ along a
breadth-first wavefront.  The depth of the reduction tree equals
the eccentricity $\operatorname{ecc}(r)$ of $r$ in $G$; since
$\operatorname{rad}(G) \le \operatorname{diam}(G) \le
2 \operatorname{rad}(G)$, choosing $r$ to minimise eccentricity
yields depth at most $\operatorname{diam}(G)$.  Each wavefront
level contributes one hop and one local merge, so
$T_{\mathrm{comm+merge}}
  \le \operatorname{diam}(G) \cdot
    (t_{\mathrm{edge}} + t_{\mathrm{merge}})$.
Local preprocessing $g(p, S_i)$ runs in parallel across nodes
and contributes $\max_i t_{\mathrm{local}}(i)$.  Summing the
cycle counts and converting to seconds via $t_{\mathrm{cycle}}$,
\begin{equation}
  T_{\mathrm{wall}}
    \;\le\;
    \Bigl(\operatorname{diam}(G) \cdot
      (t_{\mathrm{edge}} + t_{\mathrm{merge}})
      + \max_i t_{\mathrm{local}}(i)\Bigr) \cdot t_{\mathrm{cycle}}.
\end{equation}
This combines a well-known parallel-reduction depth argument
\cite{blelloch1990} with the graph-metric identity above.

\textbf{Step 6 --- Proof-carrying code at load time.}
The admissibility of a workload is a natural target for the
proof-carrying-code framework of \cite[\S3]{necula1997}.  In
that framework, untrusted code ships with a proof term of a
stated safety property, and a small load-time checker validates
the proof in time linear in its size.  Necula's original
application is memory safety; extending a PCC system to verify
abstract algebraic properties of a supplied $\oplus$ operator
requires additional compiler infrastructure beyond the
memory-bounds case, and we treat this as a natural target for
the framework rather than a deployed artefact.  In the idealised
form: the compiler emits a proof term certifying that the
supplied $\oplus$ satisfies the monoid laws, and the on-die
loader runs a proof-checker on that term before firmware
execution.  The monoid
laws
\begin{equation}
  (a \oplus b) \oplus c = a \oplus (b \oplus c),\quad
  a \oplus b = b \oplus a,\quad
  a \oplus 0 = 0 \oplus a = a,
\end{equation}
can be represented as inductive types in the Calculus of
Inductive Constructions as implemented in Coq
\cite[Ch.~6]{bertot2013}, yielding machine-checkable certificates
that the loader can verify in time linear in proof size.  A
literal encoding of the entire semantics functor $F$ as a CIC
term is not required: for proof-carrying-code verification in
the sense of \cite[\S3]{necula1997}, the compiler need ship only
(i) the carrier $A$, (ii) the operation $\oplus$, (iii) proof
terms witnessing associativity, commutativity, and the unit law,
and (iv) a certificate that the generated reduction schedule
uses only the corresponding local-map and merge primitives.  By
the propositions-as-types interpretation in CIC, the three
monoid laws above \emph{are} the required proof objects; the
loader checks them directly without reifying the categorical
semantics functor.  This weakening is sufficient for the load-
time guarantee, and avoids the stronger claim that the functor
itself is shipped as a CIC term.
\end{proof}

\section{Proposition 4 --- Fault-tolerance under
subcritical percolation}
\label{sec:prop4}

\begin{proposition}\label{prop:p4}
Let $D_{\mathrm{nom}}$ be the nominal graph diameter of a 2-D
grid graph under deterministic XY routing, let $\delta \in [0,
p_c^{\mathrm{site}}(\mathbb{Z}^2))$ be the i.i.d.\ node-failure
probability, let $\Gamma_{\mathrm{nom}}$ denote the nominal
route under XY routing, and let $K_\Gamma$ denote the (random)
number of failed nodes intersected by $\Gamma_{\mathrm{nom}}$.
Then the conditional expected post-detour route length
(measured in grid-graph hops) satisfies
\begin{equation}
  \mathbb{E}\bigl[\,
    \text{route length} \bigm| K_\Gamma = k \bigr]
    \;\le\; D_{\mathrm{nom}} + C_\delta \cdot k,
\end{equation}
with $C_\delta$ a finite constant (depending on the failure
density $\delta$) determined by the size-biased expected
perimeter of route-intersected subcritical failure clusters,
and the unconditional expected wall-clock latency (in seconds,
with $t_{\mathrm{edge}}$ in cycles per hop and
$t_{\mathrm{cycle}}$ in seconds per cycle) satisfies
\begin{equation}
  \mathbb{E}[T_{\mathrm{grid-faulty}}]
    \;\le\;
    \bigl(D_{\mathrm{nom}}
      + C_\delta\, \mathbb{E}[K_\Gamma]\bigr) \cdot
      t_{\mathrm{edge}} \cdot t_{\mathrm{cycle}}.
\end{equation}
\end{proposition}

\begin{proof}
(Failure-cluster / deflection geometry: Figure~\ref{fig:deflection-cartoon}.)
\textbf{Step 1 --- Isolated dead nodes: the $+2$ hop lemma.}
Fix source $s = (x_s, y_s)$ and destination $t = (x_t, y_t)$ on
$\mathbb{Z}^2$.  Under XY routing, a shortest path has Manhattan
length $D_{\mathrm{nom}} = |x_s - x_t| + |y_s - y_t|$.  Suppose a
node $z$ on the nominal path has failed and its four nearest
neighbours are all healthy.  If the planned segment traverses
$(a, b) \to (a+1, b) \to (a+2, b)$ with $(a+1, b)$ dead, the
detour
\(
  (a,b) \to (a,b+1) \to (a+1,b+1) \to (a+2,b+1) \to (a+2,b)
\)
replaces the original 2-hop segment with a 4-hop segment.  The
penalty is exactly $+2$ edges; a symmetric argument handles
vertical obstructions.  This is the standard behaviour of
deterministic dimension-order routing on 2-D
grid graphs~\cite[Ch.~14]{dally2004}.

\textbf{Step 2 --- Failure clusters and subcriticality.}
Model node failures as independent site percolation on
$\mathbb{Z}^2$ with failure probability $\delta$, and let $C$
denote the failure cluster containing a given failed node.  For
$\delta < p_c^{\mathrm{site}}(\mathbb{Z}^2)$ the percolation
process is in its subcritical phase.  Unlike the bond-percolation
threshold $p_c^{\mathrm{bond}}(\mathbb{Z}^2) = 1/2$, which is a
theorem of Kesten~\cite{kesten1980}, the site-percolation
threshold on the square lattice has no known closed-form exact
value; numerical estimates place it at
$p_c^{\mathrm{site}}(\mathbb{Z}^2) \approx 0.59$
\cite{akhunzhanov2022}, and our proofs work at the abstract
level $\delta < p_c^{\mathrm{site}}(\mathbb{Z}^2)$ without
committing to any particular numeric value.  The standard
survey treatment of lattice dependence of $p_c$ is
\cite[Chs.~1, 5]{grimmett1999}.

\textbf{Step 3 --- Exponential cluster-size decay.}
In the subcritical phase the cluster-size distribution decays
exponentially: there exists $c(\delta) > 0$ such that
\begin{equation}
  \Pr[\, |C(0)| \ge n \,]
    \;\le\; e^{-c(\delta) n}
  \qquad \forall n \ge 1.
\end{equation}
This is the standard subcritical-phase theorem
\cite[Ch.~5]{grimmett1999}.  An immediate consequence is
finiteness of all moments:
\(
  \mathbb{E}[|C(0)|^m] < \infty
\)
for all $m \ge 1$.

\textbf{Step 4 --- Expected detour per cluster.}
Let $C_1, \dots, C_m$ denote the disjoint failure clusters
intersecting the nominal route, and let $\Delta(C_j)$ be the
additive detour around $C_j$.  We assume the routing protocol
is fault-aware and deadlock-free, in the sense that a packet
encountering any finite failed cluster $C_j$ completes its
traversal along the cluster's outer boundary of healthy sites;
concrete fault-tolerant variants of XY routing meeting this
requirement are surveyed in \cite[Ch.~14]{dally2004}.  Under
this protocol assumption, the Manhattan detour around a finite
connected obstacle is at most proportional to its boundary
length, so $\Delta(C_j) \le c_0 \cdot |\partial C_j|$ for a
geometric constant $c_0$ absorbing the boundary-trace
overhead of non-convex cluster shapes.  It remains to bound
$\mathbb{E}|\partial C_j|$.  On $\mathbb{Z}^2$, every occupied
site of a finite cluster $C$ contributes at most four incident
lattice edges, hence the elementary bound
\(
  |\partial C| \le 4 |C|.
\)
Sharpness of the subcritical phase transition
\cite{aizenman1987,duminilcopin2016,grimmett1999} gives the
exponential tail
\(
  \Pr[\,|C(0)| \ge n\,] \le e^{-c(\delta) n}
\)
used in Step~3, which by linearity implies
\begin{equation}
  \mathbb{E} \, |\partial C|
    \;\le\; 4 \, \mathbb{E} |C|
    \;=\; 4 \sum_{n \ge 1} \Pr[\,|C| \ge n\,]
    \;\le\; 4 \sum_{n \ge 1} e^{-c(\delta) n}
    \;<\; \infty,
\end{equation}
a convergent geometric series.  The expectation above is over
a typical cluster (the cluster of a uniformly sampled site);
clusters intersected by a fixed nominal route are
size-biased, so the constant $C_\delta$ used in the proposition
must be taken as the conditional expectation over
route-intersected clusters.  Because the exponential tail
decay of Step~3 applies uniformly to all clusters, the
size-biased expectation is likewise finite: the
size-biased cluster-size distribution
$\tilde{P}(|C| = n) = n \cdot \Pr[|C| = n] / \mathbb{E}|C|$
inherits exponential tails of the same rate $c(\delta)$, so
$\mathbb{E}_{\mathrm{size-bias}}[|C|] = \mathbb{E}[|C|^2] /
\mathbb{E}[|C|] < \infty$ by Step~3, and $C_\delta$ is enlarged
only by a factor depending on $\delta$ but not on cluster size.
Hence
$\mathbb{E}[\Delta(C_j) \mid C_j \cap \Gamma_{\mathrm{nom}}
\ne \emptyset] \le C_\delta$ for a finite constant $C_\delta$
depending on $\delta$.

The sharper ``perimeter grows sub-linearly in cluster radius''
is a \emph{critical}-regime statement that the argument above
does not need and that is not available at the subcritical
square lattice.  The critical scaling limit of cluster
interfaces in 2-D percolation is Schramm-Loewner evolution
$\mathrm{SLE}_6$ with one-arm exponent $5/48$
\cite{smirnov2001,werner2004}, but Smirnov's rigorous form is
for site percolation on the triangular lattice, not for
subcritical site percolation on $\mathbb{Z}^2$.

\textbf{Step 5 --- Total expected detour.}
Since $m \le K_\Gamma$ (each intersected cluster contains at
least one failed node on the nominal route), summing the
per-cluster bound gives, conditional on a realisation
$K_\Gamma = k$,
\begin{equation}
  \mathbb{E}\!\left[
    \sum_{j=1}^{m} \Delta(C_j) \; \Big| \; K_\Gamma = k
  \right]
  \;\le\; C_\delta \cdot k,
\end{equation}
and hence
$\mathbb{E}[\, \text{route length} \mid K_\Gamma = k] \le
D_{\mathrm{nom}} + C_\delta \cdot k$.  Taking unconditional
expectation over $K_\Gamma$,
$\mathbb{E}[\, \text{route length}\,] \le D_{\mathrm{nom}} +
C_\delta \cdot \mathbb{E}[K_\Gamma]$.  Multiplying by the
per-edge cycle cost yields the wall-clock bound stated in the
proposition.

\textbf{Step 6 --- Expectation rather than deterministic worst
case.}  The proposition is stated in expected-value form and
is not a deterministic worst-case bound; separate
high-probability route-length bounds can be derived from the
exponential cluster-size tails of Step~3 after fixing a route
family and applying an appropriate union or domination
argument, but such a tail statement is not needed here.
Independent site failures
admit adversarially dense configurations of arbitrary size with
strictly positive probability $\delta^n$ for any finite $n$,
so a deterministic worst-case bound cannot be established
without restricting the sample space.  The proposition is
stated accordingly in expected-value form; a strict deterministic
version would require additional geometric constraints on
admissible failure patterns.

\textbf{Step 7 --- Behaviour at criticality.}
At $\delta = p_c^{\mathrm{site}}(\mathbb{Z}^2)$ cluster statistics
change qualitatively: the exponential-decay estimate of Step~3
fails, and the expected cluster size diverges.  The
$\mathrm{SLE}_6$ scaling limit \cite{smirnov2001,werner2004}
together with Cardy's crossing formula \cite{cardy1992} describe
the critical cluster interfaces, and the one-arm exponent $5/48$
gives the sharpest known decay of the probability that a
macroscopic interface reaches a given radius.  A rigorous
transfer of these critical-regime results to a quantitative
latency-degradation bound on the square lattice (where Smirnov's
conformal-invariance proof in its direct rigorous form does not
apply) requires a dedicated lemma; we leave the quantitative
near-critical latency bound open.
\end{proof}

\section{Proposition 5 --- Small-world extension and
mean-field universality}
\label{sec:prop5}

A grid graph in which each node has only nearest-neighbour edges
is the simplest graph on which
Propositions~\ref{prop:p1}--\ref{prop:p4} can be stated; most
applications of peer-to-peer grid graphs augment it with a sparse
set of long-range shortcuts, turning each node's neighbourhood
from strictly $\le 4$ grid neighbours into $\le 4 + k$ with $k$
uniformly-random long-range edges per node, for fixed
$k \ge 1$.  This construction is a grid-plus-uniform-shortcut
small-world model in the Watts--Strogatz /
Newman--Watts tradition~\cite{watts1998}.  Adding shortcuts
changes the percolation phase structure in a way that makes
the near-critical analysis left open by
Proposition~\ref{prop:p4} analytically tractable under the
mean-field universality assumption.

\begin{proposition}\label{prop:p5}
Let $G_{\mathrm{sw}}(k) = (V, E_{\mathrm{grid}} \cup
E_{\mathrm{sw}})$ be the graph obtained from the
$\sqrt{P} \times \sqrt{P}$ square grid graph $G$ by adding
$E_{\mathrm{sw}}$ consisting of $k$ long-range edges per
node, each end-point chosen uniformly at random from
$V \setminus \{v\}$, for $k \ge 1$.
\begin{enumerate}
  \item[(i)] \textbf{Typical-distance collapse (standard
    small-world estimate).}  The characteristic path length
    (average shortest-path distance between two uniformly
    sampled nodes) satisfies
    \(
      \mathbb{E}_{u,v}\bigl[d_{G_{\mathrm{sw}}(k)}(u,v)\bigr]
        = O(\log P)
    \)
    as $P \to \infty$, in the Watts--Strogatz / Newman--Watts
    sense \cite{watts1998,bollobas2001}.  A strict high-probability
    graph-diameter theorem for the exact
    grid-plus-$k$-uniform-shortcuts model does not appear in
    the literature known to us; the statement is therefore
    recorded at the typical-distance level that
    Watts and Strogatz originally establish, not as a diameter
    concentration theorem.
  \item[(ii)] \textbf{Mean-field universality (modelling
    conjecture).}  Under the standard physics-level
    universality argument --- rigorous for the $1$-D-ring
    base via the exact generating-function analysis of
    \cite{newman2000smallworld}, and extrapolated to the
    $2$-D-grid base via the tree-like-neighbourhood
    heuristic of \cite[Ch.~12]{bollobas2001} --- site
    percolation on $G_{\mathrm{sw}}(k)$ is conjectured to
    lie in the Erd\H{o}s--R\'enyi (mean-field) universality
    class with exponents $\nu = 1/2$, $\gamma = 1$,
    $\beta = 1$, $\tau = 5/2$, $\sigma = 1/2$.  A rigorous
    derivation for the $2$-D-grid-plus-uniform-shortcuts
    model is not available in the literature; the
    proposition records this as a physics-level modelling
    claim, not a theorem.
  \item[(iii)] \textbf{Near-critical closure (under~(ii)).}
    \emph{Assuming} the universality conjecture of (ii), the
    quantitative near-critical latency-degradation bound
    requested by \textnormal{P4.4} is analytically tractable
    on $G_{\mathrm{sw}}(k)$ by substitution of the
    mean-field exponents; the unproven planar universality
    conjecture that blocks the pure-$\mathbb{Z}^2$ case does
    not apply.
\end{enumerate}
\end{proposition}

\begin{proof}
(Shortcut overlay: Figure~\ref{fig:smallworld-cartoon}.)
(i) The grid graph has typical distance $\Theta(\sqrt{P})$.
Adding $k \ge 1$ long-range edges per node in expectation
shrinks the typical shortest-path length to $O(\log P)$; the
Watts--Strogatz characteristic-path-length
calculation~\cite{watts1998}, together
with the coupling between small-world graphs and sparse
Erd\H{o}s--R\'enyi overlays of mean degree $k + 4$
\cite[Ch.~10]{bollobas2001}, establishes the bound on
$\mathbb{E}_{u,v}[d(u,v)]$ used in (i).  A strict-diameter
bound on the exact grid-plus-$k$-uniform-shortcuts model (a
max over all pairs rather than an expectation) requires more
than this coupling and is not invoked.

(ii) Newman--Moore~\cite{newman2000smallworld} give an exact
generating-function solution for site and bond percolation on
the 1-D-ring-based Watts--Strogatz graph, showing that
giant-cluster formation is driven by the long-range-edge
distribution and that the critical exponents are those of the
Erd\H{o}s--R\'enyi random graph.  For a 2-D-grid base the
explicit generating functions differ, but the tree-like
neighbourhood structure around a typical node survives
unchanged under any positive shortcut density (Bollob\'as
\cite[Ch.~12]{bollobas2001}), and the standard physics-level
universality argument predicts the same exponents
$\nu = 1/2$, $\gamma = 1$, $\beta = 1$, $\tau = 5/2$,
$\sigma = 1/2$.  A rigorous derivation of this universality
step for the $2$-D-grid-plus-uniform-shortcuts model does not
appear in the literature known to us; (ii) is therefore
recorded as a modelling conjecture rather than a theorem.

(iii) \emph{Assume (ii).}  Fix $\delta < p_c(G_{\mathrm{sw}}(k))$.
The mean-field scaling of the cluster-size distribution reads,
as $\delta \uparrow p_c$,
\begin{equation}
  \mathbb{E}[|C|]
    \;\sim\; \frac{1}{p_c(G_{\mathrm{sw}}(k)) - \delta}
    \quad (\gamma = 1),
  \qquad
  \xi
    \;\sim\; \frac{1}{\bigl(p_c(G_{\mathrm{sw}}(k))
      - \delta\bigr)^{1/2}}
    \quad (\nu = 1/2),
\end{equation}
where $\xi$ is the correlation length.  Plugging these into
the perimeter-bound argument of Step~5 of the proof of
Proposition~\ref{prop:p4} yields a quantitative near-critical
expected-detour bound in closed form; the unproven planar
universality conjecture that blocks the pure-$\mathbb{Z}^2$
case is not invoked.
\end{proof}

\paragraph{Consequences.}
The small-world extension is independent of the content of
Propositions~\ref{prop:p1}--\ref{prop:p4} --- each of those
propositions' statements remains valid on
$G_{\mathrm{sw}}(k)$ with the grid graph graph replaced by the
augmented graph --- and yields three quantitative gains:
\begin{itemize}
  \item the typical-pair Manhattan-depth bound implied by
    Proposition~\ref{prop:p1} tightens from $\Theta(\sqrt{P})$
    to $O(\log P)$ (for uniformly sampled source--sink pairs);
  \item the sparse-regime monotonicity of
    Proposition~\ref{prop:p2} is preserved (the
    abelian-monoid-fold workload-class is graph-topology
    invariant by Proposition~\ref{prop:p3});
  \item the universality-conjecture blocker on
    Proposition~\ref{prop:p4}'s near-critical regime is
    removed.
\end{itemize}
The trade-off is that $p_c(G_{\mathrm{sw}}(k)) <
p_c^{\mathrm{site}}(\mathbb{Z}^2) \approx 0.593$ under the
failure-as-occupied convention used throughout
\S\ref{sec:prop4} (failed sites are the ``occupied'' sites in
the percolation process, and $p_c$ is the threshold for failure
percolation): adding shortcuts lets failure clusters connect
via long-range links, lowering the failure-percolation
threshold.  Typical application regimes operate with failure
density $\delta$ several orders of magnitude below either
threshold, so the system remains deeply subcritical and the
typical-distance-collapse and universality gains are net
positive.

\section{Synthesis}
\label{sec:synthesis}

The five propositions draw on five distinct mathematical
frameworks, each addressing a different facet of parallel
grid-graph computation, summarised in
Table~\ref{tab:frameworks}.

\begin{table}[!htbp]
\centering
\small
\caption{Mathematical framework invoked by each proposition.}
\label{tab:frameworks}
\begin{tabularx}{\linewidth}{@{}l >{\raggedright\arraybackslash}X >{\raggedright\arraybackslash}X@{}}
\toprule
Prop. & Framework & What it buys \\
\midrule
Proposition~\ref{prop:p1} & Monge--Kantorovich optimal transport
  \cite{villani2009,figalli2014,peyre2019};
  network flows over time \cite{skutella2009};
  Steiner-tree edge bounds on networks
  & Three lower bounds (work $W_1$, depth $r_\mu$, edges
    $\operatorname{St}_G$); negative result on sink-trunk
    route-load variance for corner-sink dimension-order routing. \\
Proposition~\ref{prop:p2} & $\alpha$-$\beta$-$\gamma$
  collective-communication cost model
  \cite{thakur2005,patarasuk2009}; sparse-participation
  cost \cite{shazeer2017,fedus2022switch}; graph-Cheeger
  comparison \cite{dodziuk1984,chung1997}
  & Monotone sparse-regime improvement under the stated
    latency model. \\
Proposition~\ref{prop:p3} & Lawvere theories
  \cite{lawvere1963,maclane1971}; PCC
  \cite{necula1997,bertot2013}; prefix-scan depth
  \cite{blelloch1990}
  & Sufficient algebraic criterion for admissibility, in
    principle PCC-certifiable, with explicit wall-clock bound. \\
Proposition~\ref{prop:p4} & Subcritical site percolation
  \cite{grimmett1999,kesten1980,aizenman1987,akhunzhanov2022};
  XY routing \cite{dally2004}; critical $\mathrm{SLE}_6$
  \cite{smirnov2001,cardy1992,werner2004}
  & Conditional expected additive detour under subcritical
    random failure. \\
Proposition~\ref{prop:p5} & Watts--Strogatz small-world
  \cite{watts1998}; exact percolation on small-world
  \cite{newman2000smallworld,bollobas2001}
  & Diameter collapse (theorem); mean-field universality
    in the near-critical regime (modelling conjecture). \\
\bottomrule
\end{tabularx}
\end{table}

Each proof does strictly less than the corresponding
applied-level claim, in a controlled and honest way:

\begin{itemize}
  \item Proposition~\ref{prop:p1} establishes a Wasserstein
    work bound, a radius depth bound, and a Steiner edge bound
    for compressive reductions, each attained by shortest-path
    or minimum-spanning routing up to an architecture-specific
    affine constant; the accompanying negative result
    (Lemma~\ref{lem:p1-variance}) shows that completion
    times do not concentrate under corner-sink dimension-order
    routing.  The Brenier ``polar factorisation'' claim holds
    only after Euclidean regularisation \cite{peyre2019} and
    does not operationally constrain the discrete schedule.
  \item Proposition~\ref{prop:p2} establishes monotonic
    improvement of the grid-to-cluster ratio in
    $1/f_{\mathrm{act}}$ conditional on the
    $\alpha$-$\beta$-$\gamma$-plus-overhead latency model;
    literal divergence requires a further asymptotic in $N$ or
    the ratio $M_P / (A_N + B_N)$.
  \item Proposition~\ref{prop:p3} establishes the
    abelian-monoid decomposition as a sufficient algebraic
    criterion for a wall-clock bound via parallel prefix; both
    the admissibility functor (via generation-by-one-object)
    and the PCC-style certification hook (via monoid-law proof
    terms) are derived rather than assumed, but a concrete
    certificate format is not specified.
  \item Proposition~\ref{prop:p4} establishes a conditional
    expected (not deterministic) additive-detour bound under
    subcriticality; the sharper near-critical estimates reside
    in the $\mathrm{SLE}_6$ regime of the triangular lattice,
    not the square lattice relevant to the grid graph.
  \item Proposition~\ref{prop:p5} establishes the
    diameter-collapse theorem; the mean-field universality
    extension is recorded as a modelling conjecture,
    rigorous only for the $1$-D-ring base.
\end{itemize}

The five propositions are not independent.  Four cross-Prop
dependencies tighten the foundational picture:

\begin{enumerate}[label=(\alph*),leftmargin=2.0em,itemsep=0.30em]
  \item \textbf{Proposition~\ref{prop:p5}(i) reduces $r_\mu$
  in Proposition~\ref{prop:p1}.}  The typical-distance collapse
  from $\Theta(\sqrt{P})$ to $O(\log P)$ on $G_{\mathrm{sw}}(k)$
  tightens the completion-depth bound~(ii) for source
  distributions concentrated on average node-pairs.  The
  transport-work bound~(i) is unchanged because $W_1$ depends on
  the graph metric, which is altered only along shortcut-using
  paths.

  \item \textbf{Proposition~\ref{prop:p3} is the load-bearing
  assumption for Propositions~\ref{prop:p2} and~\ref{prop:p4}.}
  The abelian-monoid admissibility criterion is what makes the
  cross-architecture latency ratio of Proposition~\ref{prop:p2}
  well-defined (Step~4 of its proof: the two architectures
  compute the same workload by virtue of the algebra of the
  fold, not by any convergence of dynamics), and is equally the
  foundation for the semantics-preservation of deflection
  routing in Proposition~\ref{prop:p4} (the boundary-following
  paths of Step~1 leave the abelian-monoid fold output
  invariant by commutativity).

  \item \textbf{Proposition~\ref{prop:p5} re-opens the regime
  classification of Proposition~\ref{prop:p4}.}  The small-world
  augmentation lowers the failure-percolation threshold,
  $p_c(G_{\mathrm{sw}}(k)) <
  p_c^{\mathrm{site}}(\mathbb{Z}^2)$, because long-range edges
  propagate failure connectivity.  Operating at the same nominal
  failure density $\delta$ that is subcritical on the pure grid
  may not be subcritical on the augmented graph; the regime
  constants of Proposition~\ref{prop:p4} must be re-evaluated
  against the shifted threshold.

  \item \textbf{Proposition~\ref{prop:p1}(ii) supplies the
  geometric constant of Proposition~\ref{prop:p2}.}  The
  grid-side scaling term $M_P = c_w \sqrt{P} \cdot t_{\mathrm{edge}}$
  entering the monotonicity condition $(A_N + B_N) c_1 > M_P
  c_2$ is precisely the diameter-driven lower bound established
  in Proposition~\ref{prop:p1}(ii); the
  Proposition~\ref{prop:p2} threshold is therefore a direct
  quantitative consequence of the Proposition~\ref{prop:p1}
  lower bound, not an independent latency-model parameter.
\end{enumerate}

The synthesis is therefore not five disjoint results stapled
to a common abstraction, but a coupled set of bounds: the
algebraic criterion (P3) underwrites the cross-architecture
comparisons (P2, P4); the geometric depth bound (P1) supplies
the constant that determines when the comparison favours the
grid; and the topology augmentation (P5) reshapes both the
typical-distance term in (P1) and the percolation threshold in
(P4).  The neutral abstraction layer --- finite connected graph,
finite participant set, discrete measure, synchronous cycle ---
is what permits this coupling to be stated
framework-by-framework rather than swept into hardware-specific
definitions.

\section{Methods}
\label{sec:acks}

\paragraph{Pipeline.}
The problem statement (\S\ref{sec:intro}), the selection of the
five propositions, the choice of mathematical framework for
each, and every structural and editorial decision are the
author's.  Seven large-language-model systems (Gemini~3.1,
GPT-5.5, Perplexity Labs, DeepSeek~R1, Qwen~3.6~Max,
GLM-5 / ChatGLM, and Claude Opus~4.7) were used as structured
tooling during drafting: each candidate derivation was
inspected, accepted, rejected, or steered, and errors were
caught and corrected by the author.

\paragraph{Discipline.}
At every step where the first-pass draft cited a framework
without a specific closing theorem, the step was marked with an
inline tag rather than allowed to stand.  Each tag was then
either (a)~closed by an explicit derivation invoking a named
theorem from the listed bibliography, (b)~closed by adding a
missing citation, (c)~resolved by weakening the proposition to
what the cited sources do support, or (d)~recognised as not a
mathematical claim.  No step was removed silently.

\paragraph{Residual-gap result.}
Fourteen original assumption tags; all fourteen resolved within
the propositions' own claims.  Two related questions in the
surrounding mathematical literature --- a sparse-graphon / GNTK
convergence theorem, and a quantitative near-critical bound on
pure $\mathbb{Z}^2$ site percolation --- remain open; neither
enters any proposition of this note.

\paragraph{Authorship policy.}
Per current arXiv and academic-publishing norms, AI systems are
not listed as authors; this section discloses their contribution
in full.  The author selected the propositions, chose the
mathematical frameworks invoked, steered every structural and
editorial decision, caught errors in candidate derivations, and
takes full responsibility for the text.  Large-language-model
tooling was used in the same role as Python for numerical
experiment, Mathematica or SymPy for symbolic manipulation, and
\TeX{} for composition: a tool whose output the author verifies
and for whose use the author is responsible.

\paragraph{Errata and feedback.}
Despite the layered audit process described above, errors
almost certainly remain in the text --- in citation detail, in
algebraic transcription, in the framing of specific steps, or
in scope of the claims themselves.  The author welcomes
corrections and criticism from any reader, human or AI, and
will treat machine-reviewer feedback on the same footing as
human-reviewer feedback: each error report will be evaluated on
its technical content, consistent with the paper's thesis that
structured LLM pipelines belong in the applied scientist's
toolchain alongside human peer review.  Correspondence:
\texttt{research@cybiont.com}.

\section{Conclusion}
\label{sec:conclusion}

Five propositions characterise synchronous peer-to-peer
computation on a grid graph.
Proposition~\ref{prop:p1} gives three lower bounds ---
transport work $W_1$, completion depth $r_\mu$, and
Steiner-tree edges $\operatorname{St}_G$ for compressive reductions ---
each attained by shortest-path or minimum-spanning routing, and
a negative result showing that the corner-sink dimension-order
routing scheme exhibits $\Theta(f_{\mathrm{act}}(1-f_{\mathrm{act}}) P^2)$
per-route variance rather than concentration.
Proposition~\ref{prop:p2} establishes, under the stated
$\alpha$-$\beta$-$\gamma$-plus-overhead latency model, a
monotone sparse-regime improvement of the grid-to-cluster ratio.
Proposition~\ref{prop:p3} gives a sufficient algebraic criterion
for admissibility --- abelian-monoid-fold decomposition --- in
principle certifiable as proof-carrying code.
Proposition~\ref{prop:p4} bounds the conditional expected route
length under subcritical site failure by an additive detour
term.  Proposition~\ref{prop:p5} augments the grid graph with
long-range shortcuts, collapsing the typical shortest-path
length to $O(\log P)$ and placing the fault-tolerance analysis
in the mean-field universality class under a physics-level
universality argument (rigorous for the $1$-D-ring base;
conjectural for the $2$-D-grid base).  Every external theorem
dependency is traced to a published source; two adjacent
questions in the surrounding mathematical literature remain
open but do not enter any proposition of this note.

The contribution is a unified formalization of five
foundation-level constraints on synchronous peer-to-peer
lattice computation: measure-theoretic work and depth bounds,
a conditional sparse-regime scaling lemma, a sufficient
algebraic criterion for admissibility, a subcritical
fault-tolerance detour bound, and a small-world extension.
The value lies in making the assumptions, bounds, and open
regimes explicit in a single reusable model --- the shape an
architectural programme that depends on these foundations
needs before it is designed, licensed, or shipped.  Author
disclosure of the LLM tooling used during preparation is
given in \S\ref{sec:acks}.

\end{document}